\documentclass{article}

\usepackage[centertags]{amsmath}
\usepackage{amsfonts}
\usepackage{amssymb}
\usepackage{amsthm}
\usepackage{graphicx}
\usepackage{color}

\usepackage{multirow}

\theoremstyle{plain}
\newtheorem{Ther}{Theorem}
\newtheorem{Cor}{Corollary}
\newtheorem{Lem}{Lemma}
\newtheorem*{Lem4}{Lemma 4}
\newtheorem*{Lem5}{Lemma 5}
\newtheorem*{Lem6}{Lemma 6}

\theoremstyle{definition}
\newtheorem{Def}{Definition}

\newtheorem*{Rem}{Remark}

\newcommand{\binomial}{{\rm binomial}}

\newcommand{\abs}[1]{\left| #1 \right|}
\newcommand{\set}[1]{\left\{ #1 \right\}}
\newcommand{\intSet}[1]{\left[\!\left[ #1 \right]\!\right]}
\newcommand{\pran}[1]{\left( #1 \right)}
\newcommand{\event}[1]{\left[\, #1 \,\right]}
\newcommand{\ind}[1]{\, \mathbb{I}_{\event{ #1 }} \,}
\newcommand{\prob}[1]{\mathbb{P}\set{#1}}

\newcommand{\Exp}[1]{\mathbf{E}\event{#1}}

\newcommand{\ceil}[1]{\left\lceil \, #1 \,\right\rceil}
\newcommand{\floor}[1]{\left\lfloor \, #1 \,\right\rfloor}

\begin{document}

\title{Two-way Linear Probing Revisited}

\author{Ketan Dalal\thanks{Research of the authors was supported by NSERC Grant A3456.
Emails: lucdevroye@gmail.com and emalalla@ahlia.edu.bh.}\;, Luc Devroye\footnotemark[1]\;, and
Ebrahim Malalla\footnotemark[1] \thanks{Author's current address: Department of Mathematical
Sciences, Ahlia University, P.O. Box 10878, GOSI Complex, Manama, Bahrain. The author wishes to
acknowledge the support provided by Ahlia University.}\\
School of Computer Science \\
McGill University\\
Montreal, Canada H3A 2K6}

\date{}

\maketitle

\noindent \textsc{Abstract.} We introduce linear probing hashing schemes that construct a hash
table of size $n$, with constant load factor $\alpha$, on which the worst-case unsuccessful search
time is asymptotically almost surely $O(\log \log n)$. The schemes employ two linear probe
sequences to find empty cells for the keys. Matching lower bounds on the maximum cluster size
produced by any algorithm that uses two linear probe sequences are obtained as well.

\bigskip

\noindent \textsc{Categories and Subject Descriptors:} E.2 [\textbf{Data}]: Data Storage
Representa-tions---hash-table representations; F.2.2 [\textbf{Analysis of Algorithms and Problem
Complexity}]: Non-numerical Algorithms and Problems---sorting and searching

\bigskip

\noindent \textsc{General Terms:} Algorithms, Theory

\bigskip

\noindent \textsc{Additional Keywords and Phrases:} Open addressing hashing, linear probing,
parking problem, worst-case search time, two-way chaining, multiple-choice paradigm, randomized
algorithms, witness tree, probabilistic analysis.

\section{Introduction}

In classical open addressing hashing \cite{Peterson:57}, $m$ keys are hashed sequentially and
on-line into a table of size $n > m$, (that is, a one-dimensional array with $n$ cells which we
denote by the set $\mathcal{T}=\set{0,\ldots, n-1}$), where each cell can harbor at most one key.
Each key $x$ has only one infinite probe sequence $f_i(x) \in \mathcal{T}$, for $i \in \mathbb{N}$.
During the insertion process, if a key is mapped to a cell that is already occupied by another key,
a collision occurs, and another probe is required. The probing continues until an empty cell is
reached where a key is placed. This method of hashing is pointer-free, unlike hashing with separate
chaining where keys colliding in the same cell are hashed to a separate linked list or chain. For a
discussion of different hashing schemes see \cite{Gonnet:91, Knuth:73-3, Vitter:90}.

The purpose of this paper is to design efficient open addressing hashing schemes that improve the
worst-case performance of classical linear probing where  $f_{i+1}(x) = f_i(x) + 1 \mod n$, for $i
\in \intSet{n} : = \set{1, \ldots, n}$. Linear probing is known for its good practical performance,
efficiency, and simplicity. It continues to be one of the best hash tables in practice due to its
simplicity of implementation, absence of overhead for internally used pointers, cache efficiency,
and locality of reference \cite{Janson:16, Pagh:11, Richter:15, Thorup:12}. On the other hand, the
performance of linear probing seems to degrade with high load factors $m/n$, due to a
primary-clustering tendency of one collision to cause more nearby collisions.

Our study concentrates on schemes that use two linear probe sequences to find possible hashing
cells for the keys. Each key chooses two initial cells independently and uniformly at random, with
replacement. From each initial cell, we probe linearly, and cyclically whenever the last cell in
the table is reached, to find two empty cells which we call terminal cells. The key then is
inserted into one of these terminal cells according to a fixed strategy. We consider strategies
that utilize the \emph{greedy multiple-choice paradigm} \cite{ABKU:00, Vocking:99}. We show that
some of the trivial insertion strategies with two-way linear probing have unexpected poor
performance. For example, one of the trivial strategies we study inserts each key into the terminal
cell found by the shorter probe sequence. Another simple strategy inserts each key into the
terminal cell that is adjacent to the smaller cluster, where a cluster is an isolated set of
consecutively occupied cells. Unfortunately, the performances of these two strategies are not
ideal. We prove that when any of these two strategies is used to construct a hash table with
constant load factor, the maximum unsuccessful search time is $\Omega(\log n)$, with high
probability (w.h.p.). Indeed, we prove that, w.h.p., a giant cluster of size $\Omega(\log n)$
emerges in a hash table of constant load factor, if it is constructed by a two-way linear probing
insertion strategy that always inserts any key upon arrival into the empty cell of its two initial
cells whenever one of them is empty.

Consequently, we introduce two other strategies that overcome this problem. First, we partition the
hash table into equal-sized blocks of size $\beta$, assuming $n/\beta$ is an integer. We consider
the following strategies for inserting the keys:
\begin{enumerate}
  \item[A.] Each key is inserted into the terminal cell that belongs to the least crowded block,
   i.e., the block with the least number of keys.
  \item[B.] For each block $i$, we define its weight to be the number of keys inserted into
terminal cells found by linear probe sequences whose starting locations belong to block $i$. Each
key, then, is inserted into the terminal cell found by the linear probe sequence that has started
from the block of smaller weight.
\end{enumerate}
For strategy B, we show that $\beta$ can be chosen such that for any constant load factor $\alpha
:= m/n$, the maximum unsuccessful search time is not more than $c \log_2 \log n$, w.h.p., where $c$
is a function of $\alpha$. If $\alpha < 1/2$, the same property also holds for strategy A.
Furthermore, these schemes are optimal up to a constant factor in the sense that an $\Omega(\log
\log n)$ universal lower bound holds for any strategy that uses two linear probe sequences, even if
the initial cells are chosen according to arbitrary probability distributions.

For hashing with separate chaining, one can achieve $O(\log \log n)$ maximum search time by
applying the \emph{two-way chaining} scheme \cite{ABKU:00} where each key is inserted into the
shorter chain among two chains chosen independently and uniformly at random, with replacement,
breaking ties randomly. It is proved \cite{ABKU:00, Beren:00a} that when $r=\Omega(n)$ keys are
inserted into a hash table with $n$ chains, the length of the longest chain upon termination is
$\log_2 \log n + r/n \pm O(1)$, w.h.p. Of course, this idea can be generalized to open addressing.
Assuming the hash table is partitioned into blocks of size $\beta$, we allow each key to choose two
initial cells, and hence two blocks, independently and uniformly at random, with replacement. From
each initial cell and \emph{within its block}, we probe linearly and cyclically, if necessary, to
find two empty cells; that is, whenever we reach the last cell in the block and it is occupied, we
continue probing from the first cell in the same block. The key, then, is inserted into the empty
cell that belongs to the least full block. Using the two-way chaining result, one can show that for
suitably chosen $\beta$, the maximum unsuccessful search time is $O(\log \log n)$, w.h.p. However,
this scheme uses probe sequences that are not totally linear; they are locally linear within the
blocks.

\subsection{History and Motivation}

\subsubsection*{Probing and Replacement}

Open addressing schemes are determined by the type of the probe sequence, and the replacement
strategy for resolving the collisions. Some of the commonly used probe sequences are:
\begin{enumerate}
   \item \textbf{Random Probing \cite{Morris:68}:}
For every key $x$, the infinite sequence $f_i(x)$ is assumed to be independent and uniformly
distributed over $\mathcal{T}$. That is, we require to have an infinite sequence $f_i$ of truly
uniform and independent hash functions. If for each key $x$, the first $n$ probes of the sequence
$f_i(x)$ are distinct, i.e., it is a random permutation, then it is called \textbf{uniform probing}
\cite{Peterson:57}.
  \item \textbf{Linear Probing \cite{Peterson:57}:}
For every key $x$, the first probe $f_1(x)$ is assumed to be uniform on $\mathcal{T}$, and the next
probes are defined by $f_{i+1}(x) = f_i(x) + 1 \mod n$, for $i \in \intSet{n}$. So we only require
$f_1$ to be a truly uniform hash function.
  \item \textbf{Double Probing \cite{Balbine:69}:}
For every key $x$, the first probe is $f_1(x)$, and the next probes are defined by $f_{i+1}(x) =
f_i(x) + g(x) \mod n$, for $i \in \mathbb{N}$, where $f_1$ and $g$ are truly uniform and
independent hash functions.
\end{enumerate}

Random and uniform probings are, in some sense, the idealized models \cite{Ullman:72, Yao:85}, and
their plausible performances are among the easiest to analyze; but obviously they are unrealistic.
Linear probing is perhaps the simplest to implement, but it behaves badly when the table is almost
full. Double probing can be seen as a compromise.

During the insertion process of a key $x$, suppose that we arrive at the cell $f_i(x)$ which is
already occupied by another previously inserted key $y$, that is, $f_i(x) = f_j(y)$, for some $j
\in \mathbb{N}$. Then a replacement strategy for resolving the collision is needed. Three
strategies have been suggested in the literature (see \cite{Munro:86} for other methods):
\begin{enumerate}
  \item \textsc{first come first served (fcfs) \cite{Peterson:57}:}
The key $y$ is kept in its cell, and the  key $x$ is referred to the next cell $f_{i+1}(x)$.
  \item \textsc{last come first served (lcfs) \cite{Poblete:89}:}
The key $x$ is inserted into the cell $f_i(x)$, and the key $y$ is pushed along to the next cell in
its probe sequence, $f_{j+1}(y)$.
  \item \textsc{robin hood \cite{Celis:86, Celis:85}:}
The key which travelled the furthest is inserted into the cell. That is, if $i >j$, then the key
$x$ is inserted into the cell $f_i(x)$, and the key $y$ is pushed along to the next cell
$f_{j+1}(y)$; otherwise, $y$ is kept in its cell, and the key $x$ tries its next cell $f_{i+1}(x)$.
\end{enumerate}

\subsubsection*{Average Performance}

Evidently, the performance of any open addressing scheme deteriorates when the ratio $m/n$
approaches 1, as the cluster sizes increase, where a cluster is an isolated set of consecutively
occupied cells (cyclically defined) that are bounded by empty cells. Therefore, we shall assume
that the hash table is $\alpha$-full, that is, the number of hashed keys $m = \floor{\alpha n}$,
where $\alpha \in (0,1)$ is a constant called the load factor. The asymptotic average-case
performance has been extensively analyzed for random and uniform probing \cite{Bollobas:90,
Larson:83, Morris:68, Peterson:57, Ullman:72,  Yao:85}, linear probing \cite{Knuth:63, Knuth:73-3,
Konheim:66, Mendel:80}, and double probing \cite{Balbine:69, Guibas:78, Lueker:93, Schmidt:-s,
Siegel:-s}. The expected search times were proven to be constants, more or less, depending on
$\alpha$ only. Recent results about the average-case performance of linear probing, and the limit
distribution of the construction time have appeared in \cite{Flajolet:98, Knuth:98, Viola:98}. See
also \cite{Aldous:88, Gonnet:80, Pflug:87} for the average-case analysis of linear probing for
nonuniform hash functions.

It is worth noting that the average search time of linear probing is independent of the replacement
strategy; see \cite{Knuth:73-3, Peterson:57}. This is because the insertion of any order of the
keys results in the same set of occupied cells, i.e., the cluster sizes are the same; and hence,
the total displacement of the keys---from their initial hashing locations---remains unchanged. It
is not difficult to see that this independence is also true for random and double probings. That
is, the replacement strategy does not have any effect on the average successful search time in any
of the above probings. In addition, since in linear probing the unsuccessful search time is related
to the cluster sizes (unlike random and double probings), the expected and the maximum unsuccessful
search times in linear probing are invariant to the replacement strategy.

It is known that \textsc{lcfs} \cite{Poblete:89, Poblete:97} and \textsc{robin hood}
\cite{Celis:86, Celis:85, Munro:86, Viola:98} strategies minimize the variance of displacement.
Recently, Janson \cite{Janson:03} and Viola \cite{Viola:03} have reaffirmed the effect of these
replacement strategies on the individual search times in linear probing hashing.

\subsubsection*{Worst-case Performance}

The focal point of this article, however, is the worst-case search time which is proportional to
the length of the longest probe sequence over all keys (\textsc{llps}, for short). Many results
have been established regarding the worst-case performance of open addressing.

The worst-case performance of linear probing with \textsc{fcfs} policy was analyzed by Pittel
\cite{Pittel:87}. He proved that the maximum cluster size, and hence the \textsc{llps} needed to
insert (or search for) a key, is asymptotic to $(\alpha - 1 - \log \alpha)^{-1} \log n$, in
probability. As we mentioned above, this bound holds for linear probing with \emph{any} replacement
strategy. Chassaing and Louchard \cite{Chass:02} studied the threshold of emergence of a giant
cluster in linear probing. They showed that when the number of keys $m = n - \omega(\sqrt{n})$, the
size of the largest cluster is $o(n)$, w.h.p.; however, when $m =n-o(\sqrt{n})$, a giant cluster of
size $\Theta(n)$ emerges, w.h.p.

Gonnet \cite{Gonnet:81} proved that with uniform probing and \textsc{fcfs} replacement strategy,
the expected \textsc{llps} is asymptotic to $\log_{1/\alpha} n - \log_{1/\alpha} \log_{1/\alpha} n
+ O(1)$, for $\alpha$-full tables. However, Poblete and Munro \cite{Poblete:89, Poblete:97} showed
that if random probing is combined with \textsc{lcfs} policy, then the expected \textsc{llps} is at
most $(1+o(1)) \Gamma^{-1}(\alpha n) = O(\log n / \log \log n)$, where $\Gamma$ is the gamma
function.

On the other hand, the \textsc{robin hood} strategy with random probing leads to a more striking
performance. Celis \cite{Celis:86} first proved that the expected \textsc{llps} is $O(\log n)$.
However, Devroye, Morin and Viola \ \cite{Devroye:03b} tightened the bounds and revealed that the
\textsc{llps} is indeed $\log_2 \log n \pm \Theta(1)$, w.h.p., thus achieving a double logarithmic
worst-case insertion and search times for the first time in open addressing hashing. Unfortunately,
one cannot ignore the assumption in random probing about the availability of an infinite collection
of hash functions that are sufficiently independent and behave like truly uniform hash functions in
practice. On the other side of the spectrum, we already know that \textsc{robin hood} policy does
not affect the maximum unsuccessful search time in linear probing. However, \textsc{robin hood} may
be promising with double probing.

\subsubsection*{Other Initiatives}

Open addressing methods that rely on rearrangement of keys were under investigation for many years,
see, e.g., \cite{Brent:73, Gonnet:79, Madison:80, Mallach:77, Munro:86, Rivest:78}. Pagh and Rodler
\cite{Pagh:01b} studied a scheme called cuckoo hashing that exploits the \textsc{lcfs} replacement
policy. It uses two hash tables of size $n> (1 + \epsilon)m$, for some constant $\epsilon >0$; and
two independent hash functions chosen from an $O(\log n)$-universal class---one function only for
each table. Each key is hashed initially by the first function to a cell in the first table. If the
cell is full, then the new key is inserted there anyway, and the old key is kicked out to the
second table to be hashed by the second function. The same rule is applied in the second table.
Keys are moved back and forth until a key moves to an empty location or a limit has been reached.
If the limit is reached, new independent hash functions are chosen, and the tables are rehashed.
The worst-case search time is at most two, and the amortized expected insertion time, nonetheless,
is constant. However, this scheme utilizes less than 50\% of the allocated memory, has a worst-case
insertion time of $O(\log n)$, w.h.p., and depends on a wealthy source of provably good independent
hash functions for the rehashing process. For further details see \cite{Devroye:03a, Dietzf:03,
Fotakis:03, Ostlin:03}.

The space efficiency of cuckoo hashing is significantly improved when the hash table is divided
into blocks of fixed size $b \geq 1$ and more hash functions are used to choose $k \geq 2$ blocks
for each key where each is inserted into a cell in one of its chosen blocks using the cuckoo random
walk insertion method \cite{Dietzf:05, Frieze:12, Fountoulakis:12, Fountoulakis:13, Lehman:09,
Walzer:18}. For example, it is known \cite{Dietzf:05, Lehman:09} that 89.7\% space utilization can
be achieved when $k=2$ and the hash table is partitioned into non-overlapping blocks of size $b=2$.
On the other hand, when the blocks are allowed to overlap, the space utilization improves to 96.5\%
\cite{Lehman:09, Walzer:18}. The worst-case insertion time of this generalized cuckoo hashing
scheme, however, is proven \cite{Fountoulakis:13, Frieze:11} to be polylogarithmic, w.h.p.

Many real-time static and dynamic perfect hashing schemes achieving constant worst-case search
time, and linear (in the table size) construction time and space were designed in \cite{Broder:90,
Dietzf:90, Dietzf:92b, Dietzf:92a, Dietzf:94, Fredman:84, Pagh:99, Pagh:01a}. All of these schemes,
which are based, more or less, on the idea of multilevel hashing, employ more than a constant
number of perfect hash functions chosen from an efficient universal class. Some of them even use
$O(n)$ functions.

\subsection{Our Contribution}

We design linear probing algorithms that accomplish double logarithmic worst-case search time.
Inspired by the two-way chaining algorithm \cite{ABKU:00}, and its powerful performance, we promote
the concept of open addressing hashing with two-way linear probing. The essence of the proposed
concept is based on the idea of allowing each key to generate two independent linear probe
sequences and making the algorithm decide, according to some strategy, at the end of which sequence
the key should be inserted. Formally, each input key $x$ chooses two cells independently and
uniformly at random, with replacement. We call these cells the \emph{initial hashing cells}
available for $x$. From each initial hashing cell, we start a linear probe sequence (with
\textsc{fcfs} policy) to find an empty cell where we stop. Thus, we end up with two unoccupied
cells. We call these cells the \emph{terminal hashing cells}. The question now is: into which
terminal cell should we insert the key $x$?

The insertion process of a two-way linear probing algorithm could follow one of the strategies we
mentioned earlier: it may insert the key at the end of the shorter probe sequence, or into the
terminal cell that is adjacent to the smaller cluster. Others may make an insertion decision even
before linear probing starts. In any of these algorithms, the searching process for any key is
basically the same: just start probing in both sequences alternately, until the key is found, or
the two empty cells at the end of the sequences are reached in the case of an unsuccessful search.
Thus, the maximum unsuccessful search time is at most twice the size of the largest cluster plus
two.

We study the two-way linear probing algorithms stated above, and show that the hash table,
asymptotically and almost surely, contains a giant cluster of size $\Omega(\log n)$. Indeed, we
prove that a cluster of size $\Omega(\log n)$ emerges, asymptotically and almost surely, in any
hash table of constant load factor that is constructed by a two-way linear probing algorithm that
inserts any key upon arrival into the empty cell of its two initial cells whenever one of them is
empty.

We introduce two other two-way linear probing heuristics that lead to $\Theta(\log \log n)$ maximum
unsuccessful search times. The common idea of these heuristics is the marriage between the two-way
linear probing concept and a technique we call \emph{blocking} where the hash table is partitioned
into equal-sized blocks. These blocks are used by the algorithm to obtain some information about
the keys allocation. The information is used to make better decisions about where the keys should
be inserted, and hence, lead to a more even distribution of the keys.

Two-way linear probing hashing has several advantages over other proposed hashing methods,
mentioned above: it reduces the worst-case behavior of hashing, it requires only two hash
functions, it is easy to parallelize, it is pointer-free and easy to implement, and unlike the
hashing schemes proposed in \cite{Dietzf:05, Pagh:01b}, it does not require any rearrangement of
keys or rehashing. Its maximum cluster size is $O(\log \log n)$, and its average-case performance
can be at most twice the classical linear probing as shown in the simulation results. Furthermore,
it is not necessary to employ perfectly random hash functions as it is known \cite{Pagh:11,
Richter:15, Thorup:12} that hash functions with smaller degree of universality will be sufficient
to implement linear probing schemes. See also \cite{Dietzf:92a, Karp:96, Ostlin:03, Pagh:01b,
Schmidt:-s, Siegel:95, Siegel:-s} for other suggestions on practical hash functions.

\subsubsection*{Paper Scope}

In the next section, we recall some of the useful results about the greedy multiple-choice
paradigm. We prove, in Section \ref{lowerSec}, a universal lower bound of order of $\log \log n$ on
the maximum unsuccessful search time of any two-way linear probing algorithm. We prove, in
addition, that not every two-way linear probing scheme behaves efficiently. We devote Section
\ref{blockSec} to the positive results, where we present our two two-way linear probing heuristics
that accomplish $O(\log \log n)$ worst-case unsuccessful search time. Simulation results of the
studied algorithms are summarized in Section \ref{SimulationSec}.

Throughout, we assume the following. We are given $m$ keys---from a universe set of keys
$\mathcal{U}$---to be hashed into a hash table of size $n$ such that each cell contains at most one
key. The process of hashing is sequential and on-line, that is, we never know anything about the
future keys. The constant $\alpha \in (0,1)$ is preserved in this article for the load factor of
the hash table, that is, we assume that $m = \floor{\alpha n}$. The $n$ cells of the hash table are
numbered $0, \ldots, n-1$. The linear probe sequences always move cyclically from left to right of
the hash table. The replacement strategy of all of the introduced algorithms is \textsc{fcfs}. The
\emph{insertion time} is defined to be the number of probes the algorithm performs to insert a key.
Similarly, the \emph{search time} is defined to be the number of probes needed to find a key, or
two empty cells in the case of unsuccessful search. Observe that unlike classical linear probing,
the insertion time of two-way linear probing may not be equal to the successful search time.
However, they are both bounded by the unsuccessful search time. Notice also that we ignore the time
to compute the hash functions.

\section{The Multiple-choice Paradigm}\label{MultipleSec}

Allocating balls into bins is one of the historical assignment problems \cite{Johnson:77,
Kolchin:78}. We are given $r$ balls that have to be placed into $s$ bins. The balls have to be
inserted sequentially and on-line, that is, each ball is assigned upon arrival without knowing
anything about the future coming balls. The \emph{load of a bin} is defined to be the number of
balls it contains. We would like to design an allocation process that minimizes the maximum load
among all bins upon termination. For example, in a classical allocation process, each ball is
placed into a bin chosen independently and uniformly at random, with replacement. It is known
\cite{Gonnet:81, Mitzen:96, Raab:98} that if $r =\Theta(s)$, the maximum load upon termination is
asymptotic to $\log s / \log \log s$, in probability.

On the other hand, the \emph{greedy multiple-choice allocation process}, appeared in
\cite{Eager:86, Karp:96} and studied by Azar et al.\ \cite{ABKU:00}, inserts each ball into the
least loaded bin among $d \geq 2$ bins chosen independently and uniformly at random, with
replacement, breaking ties randomly. Throughout, we will refer to this process by
\textsc{GreedyMC}$(s,r,d)$ for inserting $r$ balls into $s$ bins. Surprisingly, the maximum bin
load of \textsc{GreedyMC}$(s,s,d)$ decreases exponentially to $\log_d \log s \pm O(1)$, w.h.p.,
\cite{ABKU:00}. However, one can easily generalize this to the case $r=\Theta(s)$. It is also known
that the greedy strategy is stochastically optimal in the following sense.

\begin{Ther}[Azar et al.\ \cite{ABKU:00}] \label{ABKUTher}
Let $s, r, d \in \mathbb{N}$, where $d \geq 2$, and $r = \Theta(s)$. Upon termination of
\textsc{GreedyMC}$(s,r,d)$, the maximum bin load is $\log_d \log s \pm O(1)$, w.h.p. Furthermore,
the maximum bin load of any on-line allocation process that inserts $r$ balls sequentially into $s$
bins where each ball is inserted into a bin among $d$ bins chosen independently and uniformly at
random, with replacement, is at least $\log_d \log s - O(1)$, w.h.p.
\end{Ther}

Berenbrink et al.\ \cite{Beren:00a} extended Theorem \ref{ABKUTher} to the
heavily loaded case where $r \gg s$, and recorded the following tight
result.

\begin{Ther}[Berenbrink et al.\ \cite{Beren:00a}] \label{BerenTher1}
There is a constant $C>0$ such that for any integers $r \geq s >0$,
and $d \geq 2$, the maximum bin load upon termination of the process
\textsc{GreedyMC}$(s,r, d)$ is $\log_d \log s + r/s \pm C$, w.h.p.
\end{Ther}

Theorem \ref{BerenTher1} is a crucial result that we have used to derive our
results, see Theorems \ref{localLinearTher} and \ref{decideTher}. It states
that the deviation from the average bin load which is $\log_d \log s$ stays
unchanged as the number of balls increases.

V\"ocking \cite{Vocking:99, Vocking:01} demonstrated that it is possible to
improve the performance of the greedy process, if non-uniform distributions
on the bins and a tie-breaking rule are carefully chosen. He suggested the
following variant which is called \emph{Always-Go-Left}. The bins are
numbered from 1 to $n$. We partition the $s$ bins into $d$ groups of almost
equal size, that is, each group has size $\Theta(s/d)$. We allow each ball
to select upon arrival $d$ bins independently at random, but the $i$-th bin
must be chosen uniformly from the $i$-th group. Each ball is placed on-line,
as before, in the least full bin, but upon a tie, the ball is always placed
in the leftmost bin among the $d$ bins. We shall write
\textsc{LeftMC}$(s,r,d)$ to refer to this process. V\"ocking
\cite{Vocking:99} showed that if $r = \Theta(s)$, the maximum load of
\textsc{LeftMC}$(s,r,d)$ is $\log \log s / (d \log \phi_d) + O(1)$, w.h.p.,
where $\phi_d$ is a constant related to a generalized Fibonacci sequence.
For example, the constant $\phi_2= 1.61...$ corresponds to the well-known
golden ratio, and $\phi_3 = 1.83...$. In general, $\phi_2< \phi_3 < \phi_4 <
\cdots < 2$, and $\lim_{d \to \infty} \phi_d= 2$. Observe the improvement on
the performance of \textsc{GreedyMc}$(s,r,d)$, even for $d=2$. The maximum
load of \textsc{LeftMC}$(s,r,2)$ is $0.72... \times \log_2 \log s + O(1)$,
whereas in \textsc{GreedyMC}$(s,r,2)$, it is $\log_2 \log s + O(1)$. The
process \textsc{LeftMC}$(s,r,d)$ is also optimal in the following sense.

\begin{Ther}[V\"ocking \cite{Vocking:99}] \label{VockingTher}
Let $r, s, d \in \mathbb{N}$, where $d\geq 2$, and $r=\Theta(s)$. The
maximum bin load of of \textsc{LeftMC}$(s,r,d)$ upon termination is $\log
\log s / (d \log \phi_d) \pm O(1)$, w.h.p. Moreover, the maximum bin load of
any on-line allocation process that inserts $r$ balls sequentially into $s$
bins where each ball is placed into a bin among $d$ bins chosen according to
arbitrary, not necessarily independent, probability distributions defined on
the bins is at least $\log \log s / (d \log \phi_d) - O(1)$, w.h.p.
\end{Ther}

Berenbrink et al.\ \cite{Beren:00a} studied the heavily loaded case and
recorded the following theorem.

\begin{Ther}[Berenbrink et al.\ \cite{Beren:00a}] \label{BerenbTher2}
There is a constant $A >0$ such that for any integers $r \geq s >0$,
and $d \geq 2$, the maximum bin load upon termination of the process
\textsc{LeftMC}$(s,r, d)$ is $\log \log s / (d \log \phi_d) + r/s
\pm A$, w.h.p.
\end{Ther}

For other variants and generalizations of the multiple-choice
paradigm see \cite{Adler:98, Adler:95, Beren:00b, Czumaj:01,
Mitzen:99a, Stemann:96}. The paradigm has been used to derive many
applications, e.g., in load balancing, circuit routing, IP address
lookups, and computer graphics \cite{Broder:01, Mitzen:96,
Mitzen:99b, Mitzen:00, Wu:02}.

\section{Life is not Always Good!}\label{lowerSec}

We prove here that the idea of two-way linear probing alone is not always sufficient to pull off a
plausible hashing performance. We prove that a large group of two-way linear probing algorithms
have an $\Omega(\log n)$ lower bound on their worst-case search time. To avoid any ambiguity, we
consider this definition.

\begin{Def}\label{def2wlinear}
A two-way linear probing algorithm is an open addressing hashing algorithm that inserts keys into
cells using a certain strategy and does the following upon the arrival of each key:
\begin{enumerate}
  \item It chooses two initial hashing cells independently and uniformly
at random, with replacement.
  \item Two terminal (empty) cells are then found by linear probe sequences
starting from the initial cells.
  \item The key is inserted into one of these terminal cells.
\end{enumerate}
\end{Def}

To be clear, we give two examples of inefficient two-way linear probing algorithms. Our first
algorithm places each key into the terminal cell discovered by the shorter probe sequence. More
precisely, once the key chooses its initial hashing cells, we start two linear probe sequences. We
proceed, sequentially and alternately, one probe from each sequence until we find an empty
(terminal) cell where we insert the key. Formally, let $f, \, g :\mathcal{U} \to \set{0, \ldots,
n-1}$ be independent and truly uniform hash functions. For $x \in \mathcal{U}$, define the linear
sequence $f_1(x) = f(x)$, and $f_{i+1}(x) = f_i(x) + 1 \mod n$, for $i \in \intSet{n}$; and
similarly define the sequence $g_i(x)$. The algorithm, then, inserts each key $x$ into the first
unoccupied cell in the following probe sequence: $f_1(x), \, g_1(x), \, f_2(x), \, g_2(x), \,
f_3(x), \, g_3(x), \ldots$. We denote this algorithm that hashes $m$ keys into $n$ cells by
\textsc{ShortSeq}$(n,m)$, for the shorter sequence.

\begin{figure}[htbp]
\begin{center}
  \scalebox{1}{\input{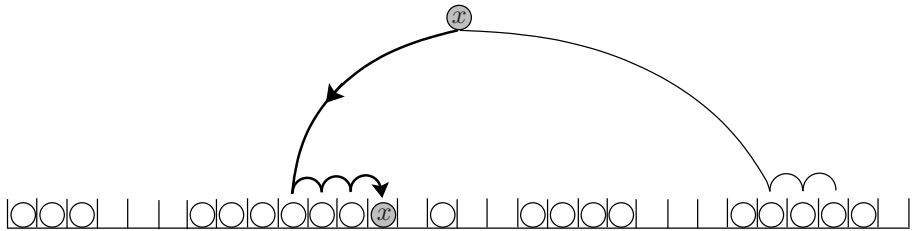}}
  \caption{An illustration of algorithm \textsc{ShortSeq}$(n,m)$ in terms of balls (keys)
and bins (cells). Each ball is inserted into the empty bin found by the shorter
sequence.}\label{figShortSeq}
\end{center}
\end{figure}

The second algorithm inserts each key into the empty (terminal) cell that is the right neighbor of
the smaller cluster among the two clusters containing the initial hashing cells, breaking ties
randomly. If one of the initial cells is empty, then the key is inserted into it, and if both of
the initial cells are empty, we break ties evenly. Recall that a cluster is a group of
consecutively occupied cells whose left and right neighbors are empty cells. This means that one
can compute the size of the cluster that contains an initial hashing cell by running two linear
probe sequences in opposite directions starting from the initial cell and going to the empty cells
at the boundaries. So practically, the algorithm uses four linear probe sequences. We refer to this
algorithm by \textsc{SmallCluster}$(n,m)$ for inserting $m$ keys into $n$ cells.

\begin{figure}[htbp]
\begin{center}
  \scalebox{1}{\input{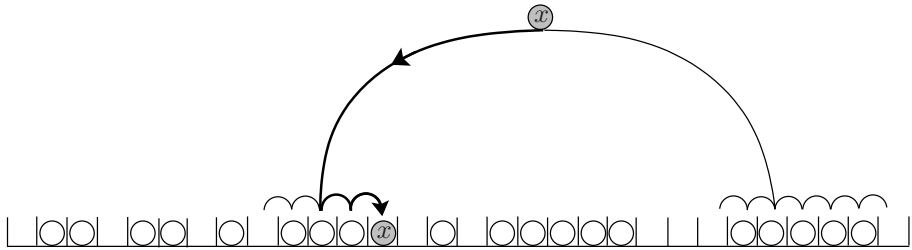}}
  \caption{Algorithm \textsc{SmallCluster}$(n,m)$ inserts each key into
the empty cell adjacent to the smaller cluster, breaking ties randomly. The size of the clusters is
determined by probing linearly in both directions.}\label{figSmallCluster}
\end{center}
\end{figure}

Before we show that these algorithms produce large clusters, we shall record a lower bound that
holds for any two-way linear probing algorithm.

\subsection{Universal Lower Bound}

The following lower bound holds for any two-way linear probing hashing
scheme, in particular, the ones that are presented in this article.

\begin{Ther}\label{universalLower1}
Let $n \in \mathbb{N}$, and  $m=\floor{\alpha n}$, where $\alpha \in (0,1)$
is a constant. Let $\mathbf{A}$ be any two-way linear probing algorithm that
inserts $m$ keys into a hash table of size $n$. Then upon termination of
$\mathbf{A}$, w.h.p., the table contains a cluster of size of at least
$\log_2 \log n - O(1)$.
\end{Ther}
\begin{proof}
Imagine that we have a bin associated with each cell in the hash table.
Recall that for each key $x$, algorithm $\mathbf{A}$ chooses two initial
cells, and hence two bins, independently and uniformly at random, with
replacement. Algorithm $\mathbf{A}$, then, probes linearly to find two
(possibly identical) terminal cells, and inserts the key $x$ into one of
them. Now imagine that after the insertion of each key $x$, we also insert a
ball into the bin associated with the initial cell from which the algorithm
started probing to reach the terminal cell into which the key $x$ was
placed. If both of the initial cells lead to the same terminal cell, then we
break the tie randomly. Clearly, if there is a bin with $k$ balls, then
there is a cluster of size of at least $k$, because the $k$ balls represent
$k$ distinct keys that belong to the same cluster. However, Theorem
\ref{ABKUTher} asserts that the maximum bin load upon termination of
algorithm $\mathbf{A}$ is at least $\log_2 \log n - O(1)$, w.h.p.
\end{proof}

The above lower bound is valid for all algorithms that satisfy
Definition \ref{def2wlinear}. A more general lower bound can be
established on all open addressing schemes that use two linear probe
sequences where the initial hashing cells are chosen according to
some (not necessarily uniform or independent) probability
distributions defined on the cells. We still assume that the probe
sequences are used to find two (empty) terminal hashing cells, and
the key is inserted into one of them according to some strategy. We
call such schemes \emph{nonuniform two-way linear probing}. The
proof of the following theorem is basically similar to Theorem
\ref{universalLower1}, but by using instead V\"ocking's lower bound
as stated in Theorem \ref{VockingTher}.

\begin{Ther}\label{universalLower2}
Let $n \in \mathbb{N}$, and  $m=\floor{\alpha n}$, where $\alpha \in (0,1)$
is a constant.  Let $\mathbf{A}$ be any nonuniform two-way linear probing
algorithm that inserts $m$ keys into a hash table of size $n$ where the
initial hashing cells are chosen according to some probability
distributions. Then the maximum cluster size produced by $\mathbf{A}$, upon
termination, is at least $\log \log n / (2 \log \phi_2) - O(1)$, w.h.p.
\end{Ther}

\subsection{Algorithms that Behave Poorly} \label{badAlgSec}

We characterize some of the inefficient two-way linear probing algorithms.
Notice that the main mistake in algorithms \textsc{ShortSeq}$(n,m)$ and
\textsc{SmallCluster}$(n,m)$ is that the keys are allowed to be inserted
into empty cells even if these cells are very close to some giant clusters.
This leads us to the following theorem.

\begin{Ther} \label{emptyCellsTher}
Let $\alpha \in (0,1)$ be constant. Let $\mathbf{A}$ be a two-way
linear probing algorithm that inserts $m=\floor{\alpha n}$ keys into
$n$ cells such that whenever a key chooses an empty and an occupied
initial cells, the algorithm inserts the key into the empty one. Then
algorithm $\mathbf{A}$ produces a giant cluster of size $\Omega(\log
n)$, w.h.p.
\end{Ther}

To prove the theorem, we need to recall the following.

\begin{Def}[See, e.g., \cite{Dubhashi:98}]
Any non-negative random variables $X_1, \ldots, X_n$ are said to be
negatively associated, if for every disjoint index subsets $I, J
\subseteq \intSet{n}$, and for any functions $f :
\mathbb{R}^{\abs{I}} \to \mathbb{R}$, and $g : \mathbb{R}^{\abs{J}}
\to \mathbb{R}$ that are both non-decreasing or both non-increasing
(componentwise), we have
\[ \Exp{f(X_i, i \in I) \, g(X_j, j \in J) }
        \leq \Exp{f(X_i, i \in I) } \,
                     \Exp{g(X_j, j \in J) } \, .    \]
\end{Def}

Once we establish that $X_1, \ldots, X_n$ are negatively associated, it
follows, by considering inductively the indicator functions, that
\[ \prob{X_1 < x_1, \ldots, X_n, < x_n  } \leq \prod_{i=1}^n \prob{X_i < x_i } \, . \]
The next lemmas, which are proved in \cite{Dubhashi:98, Esary:67,
Joag:83}, provide some tools for establishing the negative
association.

\begin{Lem}[Zero-One Lemma] \label{lemNA1}
Any binary random variables $X_1, \ldots, X_n$ whose sum is one are
negatively associated.
\end{Lem}

\begin{Lem} \label{lemNA2}
If $\set{X_1, \ldots, X_n}$ and $\set{Y_1, \ldots, Y_m}$ are independent
sets of negatively associated random variables, then the union $\set{X_1,
\ldots, X_n, Y_1, \ldots, Y_m}$ is also a set of negatively associated
random variables.
\end{Lem}

\begin{Lem} \label{lemNA3}
Suppose that $X_1, \ldots, X_n$ are negatively associated. Let $I_1,
\ldots, I_k \subseteq \intSet{n}$ be disjoint index subsets, for
some positive integer $k$. For $j \in \intSet{k}$, let $h_j :
\mathbb{R}^{\abs{I_j}} \to \mathbb{R}$ be non-decreasing functions,
and define $Z_j = h_j (X_i, i \in I_j)$. Then the random variables
$Z_1, \ldots, Z_k$ are negatively associated. In other words,
non-decreasing functions of disjoint subsets of negatively
associated random variables are also negatively associated. The same
holds if $h_j$ are non-increasing functions.
\end{Lem}

Throughout, we write $\binomial(n, p)$ to denote a binomial random
variable with parameters $n \in \mathbb{N}$ and $p \in [0,1]$.

\paragraph{Proof of Theorem \ref{emptyCellsTher}.}
Let $\beta = \floor{b \log_a n}$ for some positive constants $a$ and
$b$ to be defined later, and without loss of generality, assume that
$N := n/ \beta$ is an integer. Suppose that the hash table is
divided into $N$ disjoint blocks, each of size $\beta$. For $i \in
\intSet{N}$, let $\mathfrak{B}_i = \set{\beta( i-1)+1, \ldots, \beta
i}$ be the set of cells of the $i$-th block, where we consider the
cell numbers in a circular fashion. We say that a cell $j \in
\intSet{n}$ is ``covered" if there is a key whose first initial
hashing cell is the cell $j$ and its second initial hashing cell is
an occupied cell. A block is \emph{covered} if all of its cells are
covered. Observe that if a block is covered then it is fully
occupied. Thus, it suffices to show that there would be a covered
block, w.h.p.

For $i \in \intSet{N}$, let $Y_i$ be the indicator that the $i$-th
block is covered. The random variables $Y_1, \ldots, Y_N$ are
negatively associated which can been seen as follows. For $j \in
\intSet{n}$ and $t \in \intSet{m}$, let $X_j(t)$ be the indicator
that the $j$-th cell is covered by the $t$-th key, and set $X_0(t)
:= 1 - \sum_{j=1}^n X_j(t)$. Notice that the random variable
$X_0(t)$ is binary. The zero-one Lemma asserts that the binary
random variables $X_0(t), \ldots, X_n(t)$ are negatively associated.
However, since the keys choose their initial hashing cells
independently, the random variables $X_0(t), \ldots, X_n(t)$ are
mutually independent from the random variables $X_0(t'), \ldots,
X_n(t')$, for any distinct $t, t' \in \intSet{m}$. Thus, by Lemma
\ref{lemNA2}, the union $\cup_{t=1}^m \set{X_0(t), \ldots, X_n(t)}$
is a set of negatively associated random variables. The negative
association of the $Y_i$ is assured now by Lemma \ref{lemNA3} as
they can be written as non-decreasing functions of disjoint subsets
of the indicators $X_j(t)$. Since the $Y_i$ are negatively
associated and identically distributed, then
\[ \prob{Y_1 = 0, \ldots, Y_N =0} \leq \prob{Y_1=0} \times \cdots \times \prob{Y_N=0}
      \leq \exp \pran{- N \prob{Y_1 = 1 }} \, . \]
Thus, we only need to show that $N \prob{Y_1 = 1}$ tends to infinity
as $n$ goes to infinity. To bound the last probability, we need to
focus on the way the first block $\mathfrak{B}_1 = \set{1, 2,
\ldots, \beta}$ is covered. For $j \in \intSet{n}$, let $t_j$ be the
smallest $t \in \intSet{m}$ such that $X_j(t) =1$ (if such exists),
and $m+1$ otherwise. We say that the first block is ``covered in
order" if and only if $1 \leq t_1 < t_2 < \cdots < t_{\beta} \leq
m$. Since there are $\beta!$ orderings of the cells in which they
can be covered (for the first time), we have
\[ \prob{Y_1 =1}
     = \beta! \; \prob{\mathfrak{B}_1 \text{is covered in order}} \, .  \]
For $t \in \intSet{m}$, let $M(t) = 1$ if block $\mathfrak{B}_1$ is full before the insertion of
the $t$-th key, and otherwise be the minimum $i \in \mathfrak{B}_1$ such that the cell $i$ has not
been covered yet. Let $A$ be the event that, for all $t \in \intSet{m}$, the first initial hashing
cell of the $t$-th key is either cell $M(t)$ or a cell outside $\mathfrak{B}_1$. Define the random
variable $W := \sum_{t=1}^m W_t$, where $W_t$ is the indicator that the $t$-th key covers a cell in
$\mathfrak{B}_1$. Clearly, if $A$ is true and $W \geq \beta$, then the first block is covered in
order. Thus,
\[ \prob{Y_1 =1} \geq \beta! \; \prob{\event{W \geq \beta} \cap A}
         = \beta! \; \prob{A} \prob{W \geq \beta \, | \, A} \, . \]
However, since the initial hashing cells are chosen independently
and uniformly at random, then for $n$ chosen large enough, we have
\[ \prob{A} \geq \pran{1 - \frac{\beta}{n} }^m  \geq e^{- 2 \beta } \, , \]
and for $t \geq \ceil{m/2}$,
\[ \prob{W_t = 1 \, | \, A} = \frac{1}{n-\beta+1} \cdot \frac{t-1}{n}
     \geq \frac{\alpha}{4 n}  \, . \]
Therefore, for $n$ sufficiently large, we get
\begin{eqnarray*}
  N \prob{Y_1 = 1}
  &\geq& N \beta! \, e^{- 2 \beta } \,
            \prob{\binomial\pran{\ceil{m/2}, \alpha/(4n)} \geq \beta } \\
   &\geq& N \beta! \, e^{- 2 \beta } \,
            \frac{(m/2 -\beta)^\beta}{\beta!}
               \pran{\frac{\alpha}{4n}}^{\beta}
             \pran{1-\frac{\alpha}{4n}}^n \\
  &\geq&     N \pran{ \frac{ \alpha^2}{8 e^2}}^{\beta}
         \pran{1-\frac{2\beta}{m}}^{\beta} \pran{1-\frac{1}{4n}}^n\\
  &\geq&  \frac{n}{4 \beta} \pran{ \frac{ \alpha^2}{8 e^2}}^{\beta} \, ,
\end{eqnarray*}
which goes to infinity as $n$ approaches infinity whenever  $a = 8 e^2 /\alpha^2$ and $b$ is any
positive constant less than 1. \hfill$\square$ \\[-4pt]

Clearly, algorithms \textsc{ShortSeq}$(n,m)$ and \textsc{SmallCluster}$(n,m)$ satisfy the condition
of Theorem \ref{emptyCellsTher}. So this corollary follows.

\begin{Cor}
Let $n \in \mathbb{N}$, and $m = \floor{\alpha n}$, where $\alpha \in (0,1)$ is constant. The size
of the largest cluster generated by algorithm \textsc{ShortSeq}$(n,m)$ is $\Omega(\log n)$, w.h.p.
The same result holds for algorithm \textsc{SmallCluster}$(n,m)$.
\end{Cor}

\section{Hashing with Blocking} \label{blockSec}

To overcome the problems of Section \ref{badAlgSec}, we introduce
\emph{blocking}. The hash table is partitioned into equal-sized
disjoint blocks of cells. Whenever a key has two terminal cells, the
algorithm considers the information provided by the blocks, e.g.,
the number of keys it harbors, to make a decision. Thus, the
blocking technique enables the algorithm to avoid some of the bad
decisions the previous algorithms make. This leads to a more
controlled allocation process, and hence, to a more even
distribution of the keys. We use the blocking technique to design
two two-way linear probing algorithms, and an algorithm that uses
linear probing locally within each block. The algorithms are
characterized by the way the keys pick their blocks to land in. The
worst-case performance of these algorithms is analyzed and proven to
be $O(\log \log n)$, w.h.p.

Note also that (for insertion operations only) the algorithms require a counter with each block,
but the extra space consumed by these counters is asymptotically negligible. In fact, we will see
that the extra space is $O(n/\log \log n)$ in a model in which integers take $O(1)$ space, and at
worst $O(n \log \log \log n /\log \log n)=o(n)$ units of memory, w.h.p., in a bit model.

Since the block size for each of the following algorithms is
different, we assume throughout and without loss of generality, that
whenever we use a block of size $\beta$, then $n/\beta$ is an
integer. Recall that the cells are numbered $0, \ldots, n-1$, and
hence, for $i \in \intSet{n /\beta}$, the $i$-th block consists of
the cells $(i-1) \beta, \ldots, i \beta -1$. In other words, the
cell $k \in \set{0, \ldots, n-1}$ belongs to block number
$\lambda(k) := \floor{k/\beta} + 1$.

\subsection{Two-way Locally-Linear Probing}

As a simple example of the blocking technique, we present the following
algorithm which is a trivial application of the two-way chaining scheme
\cite{ABKU:00}. The algorithm does not satisfy the definition of two-way
linear probing as we explained earlier, because the linear probes are
performed within each block and not along the hash table. That is, whenever
the linear probe sequence reaches the right boundary of a block, it
continues probing starting from the left boundary of the same block.

The algorithm partitions the hash table into disjoint blocks each of size $\beta_1(n)$, where
$\beta_1(n)$ is an integer to be defined later. We save with each block its load, that is, the
number of keys it contains, and keep it updated whenever a key is inserted in the block. For each
key we choose two initial hashing cells, and hence two blocks, independently and uniformly at
random, with replacement. From the initial cell that belongs to the least loaded block, breaking
ties randomly, we probe linearly and cyclically \emph{within the block} until we find an empty cell
where we insert the key. If the load of the block is $\beta_1$, i.e., it is full, then we check its
right neighbor block and so on, until we find a block that is not completely full. We insert the
key into the first empty cell there. Notice that only one probe sequence is used to insert any key.
The search operation, however, uses two probe sequences as follows. First, we compute the two
initial hashing cells. We start probing linearly and cyclically \emph{within} the two (possibly
identical) blocks that contain these initial cells. If both probe sequences reach empty cells, or
if one of them reaches an empty cell and the other one finishes the block without finding the key,
we declare the search to be unsuccessful. If both blocks are full and the probe sequences
completely search them without finding the key, then the right neighbors of these blocks
(cyclically speaking) are searched sequentially in the same way mentioned above, and so on. We will
refer to this algorithm by \textsc{LocallyLinear}$(n,m)$ for inserting $m$ keys into $n$ cells. We
show next that $\beta_1$ can be defined such that none of the blocks are completely full, w.h.p.
This means that whenever we search for any key, most of the time, we only need to search linearly
and cyclically the two blocks the key chooses initially.

\begin{Ther} \label{localLinearTher}
Let $n \in \mathbb{N}$, and $m = \floor{\alpha n}$, where $\alpha \in (0,1)$
is a constant. Let $C$ be the constant defined in Theorem \ref{BerenTher1},
and define
\[ \beta_1 (n) := \floor{ \frac{\log_2 \log n + C}{1-\alpha} +1 } \, .  \]
Then, w.h.p., the maximum unsuccessful search time of
\textsc{LocallyLinear}$(n,m)$ with blocks of size $\beta_1$ is at most
$2\beta_1$, and the maximum insertion time is at most $\beta_1 -1$.
\end{Ther}
\begin{proof}
Notice the equivalence between algorithm \textsc{LocallyLinear}$(n,m)$ and
the allocation process \textsc{GreedyMC}$(n/\beta_1, m, 2)$ where $m$ balls
(keys) are inserted into $n/\beta_1$ bins (blocks) by placing each ball into
the least loaded bin among two bins chosen independently and uniformly at
random, with replacement. It suffices, therefore, to study the maximum bin
load of \textsc{GreedyMC}$(n/\beta_1, m, 2)$ which we denote by $L_n$.
However, Theorem \ref{BerenTher1} says that w.h.p.,
\[ L_n  \leq  \log_2 \log n + C  + \alpha \beta_1
              <  (1-\alpha)\beta_1 + \alpha \beta_1  = \beta_1 \, .
\]
and similarly,
\[ L_n  \geq \log_2 \log n + \alpha \beta_1  - C
           >  \frac{\log_2 \log n + C}{1 - \alpha}  - 2 C
         \geq \beta_1 -  2 C - 1 \, .
\]
\end{proof}

\subsection{Two-way Pre-linear Probing:
algorithm D{\small{ECIDE}}F{\small{IRST}}}

In the previous two-way linear probing algorithms, each input key initiates linear probe sequences
that reach two terminal cells, and then the algorithms decide in which terminal cell the key should
be inserted. The following algorithm, however, allows each key to choose two initial hashing cells,
and then decides, according to some strategy, which initial cell should start a linear probe
sequence to find a terminal cell to harbor the key. So, technically, the insertion process of any
key uses only one linear probe sequence, but we still use two sequences for any search.

\begin{figure}[htbp]
\begin{center}
  \scalebox{0.95}{\input{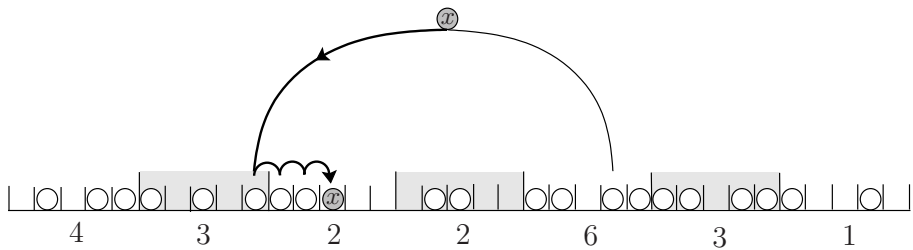}}
  \caption{An illustration of algorithm \textsc{DecideFirst}$(n,m)$.
The hash table is divided into blocks of size $\beta_2$. The number
under each block is its weight. Each key decides first to land into
the block of smaller weight, breaking ties randomly, then probes
linearly to find its terminal cell.}\label{figDecideFirst}
\end{center}
\end{figure}

Formally, we describe the algorithm as follows. Let $\alpha \in
(0,1)$ be the load factor. Partition the hash table into blocks of
size $\beta_2 (n)$, where $\beta_2 (n)$ is an integer to be defined
later. Each key $x$ still chooses, independently and uniformly at
random, two initial hashing cells, say $I_x$ and $J_x$, and hence,
two blocks which we denote by $\lambda(I_x)$ and $\lambda(J_x)$. For
convenience, we say that the key $x$ has \emph{landed} in block $i$,
if the linear probe sequence used to insert the key $x$ has started
(from the initial hashing cell available for $x$) in block $i$.
Define the weight of a block to be the number of keys that have
landed in it. We save with each block its weight, and keep it
updated whenever a key lands in it. Now, upon the arrival of key
$x$, the algorithm allows $x$ to land into the block among
$\lambda(I_x)$ and $\lambda(J_x)$ of smaller weight, breaking ties
randomly. Whence, it starts probing linearly from the initial cell
contained in the block until it finds a terminal cell into which the
key $x$ is placed. If, for example, both $I_x$ and $J_x$ belong to
the same block, then $x$ lands in $\lambda(I_x)$, and the linear
sequence starts from an arbitrarily chosen cell among $I_x$ and
$J_x$. We will write \textsc{DecideFirst}$(n,m)$ to refer to this
algorithm for inserting $m$ keys into $n$ cells.

In short, the strategy of \textsc{DecideFirst}$(n,m)$ is: land in
the block of smaller weight, walk linearly, and insert into the
first empty cell reached. The size of the largest cluster produced
by the algorithm is $\Theta(\log \log n)$. The performance of this
hashing technique is described in Theorem \ref{decideTher}:

\begin{Ther} \label{decideTher}
Let $n \in \mathbb{N}$, and $m = \floor{\alpha n}$, where $\alpha \in (0,1)$
is a constant. There is a constant $\eta >0$ such that if
\[ \beta_2 (n)
     := \ceil{\frac{(1+\sqrt{2 - \alpha})}{\sqrt{2- \alpha}(1-\alpha)}
               (\log_2 \log n + \eta)  } \, ,
\]
then, w.h.p., the worst-case unsuccessful search time of algorithm
\textsc{DecideFirst}$(n,m)$ with blocks of size $\beta_2$ is at most $\xi_n
:= 12(1-\alpha)^{-2} (\log_2 \log n +\eta)$, and the maximum insertion time
is at most $\xi_n/2$.
\end{Ther}
\begin{proof}
Assume first that \textsc{DecideFirst}$(n,m)$ is applied to a hash table with blocks of size
$\beta= \ceil{b (\log_2 \log n + \eta)}$, and that $n/\beta$ is an integer, where $b =
(1+\epsilon)/(1-\alpha)$, for some arbitrary constant $\epsilon>0$. Consider the resulting hash
table after termination of the algorithm. Let $M \geq 0$ be the maximum number of consecutive
blocks that are fully occupied. Without loss of generality, suppose that these blocks start at
block $i+1$, and let $S= \set{i, i+1, \ldots, i + M}$ represent these full blocks in addition to
the left adjacent block that is not fully occupied (Figure \ref{fullblocks}).

\begin{figure}[htbp]
\begin{center}
  \scalebox{0.8}{\input{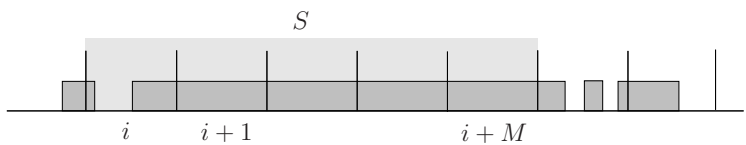}}
  \caption{A portion of the hash table showing the largest cluster,
and the set $S$ which consists of the full consecutive blocks and their left
neighbor.}\label{fullblocks}
\end{center}
\end{figure}

Notice that each key chooses two cells (and hence, two possibly
identical blocks) independently and uniformly at random. Also, any
key always lands in the block of smaller weight. Since there are $n
/ \beta$ blocks, and $\floor{\alpha n}$ keys, then by Theorem
\ref{BerenTher1}, there is a constant $C>0$ such that the maximum
block weight is not more than $\lambda_n := (\alpha b + 1) \log_2
\log n +\alpha b \eta + \alpha + C$, w.h.p. Let $A_n$ denote the
event that the maximum block weight is at most $\lambda_n$. Let $W$
be the number of keys that have landed in $S$, i.e., the total
weight of blocks contained in $S$. Plainly, since block $i$ is not
full, then all the keys that belong to the $M$ full blocks have
landed in $S$. Thus, $W \geq M b (\log_2 \log n + \eta)$,
deterministically. Now, clearly, if we choose $\eta = C+ \alpha$,
then the event $A_n$ implies that $(M+1) (\alpha b + 1) \geq M b$,
because otherwise, we have
\[ W \leq (M+1) (\alpha b + 1) \pran{ \log_2 \log n +
      \frac{\alpha b \eta + \alpha + C}{\alpha b + 1 } }
       <  M b (\log_2 \log n + \eta)   \, ,
\]
which is a contradiction. Therefore, $A_n$ yields that
\[ M \leq \frac{\alpha b + 1}{(1 -\alpha )b-1}
     \leq \frac{1 + \epsilon \alpha}{\epsilon (1-\alpha)} \, .
\]
Recall that $(\alpha b + 1) < b=(1+\epsilon)/(1-\alpha)$. Again, since block
$i$ is not full, the size of the largest cluster is not more than the total
weight of the $M+2$ blocks that cover it. Consequently, the maximum cluster
size is, w.h.p., not more than
\[ (M+2)  (\alpha b + 1) (\log_2 \log n +\eta )
       \leq \frac{\psi(\epsilon)}{(1-\alpha)^2} (\log_2 \log n +\eta)  \, ,
\]
where $\psi(\epsilon) := 3 - \alpha + (2-\alpha)\epsilon + 1/\epsilon$. Since $\epsilon$ is
arbitrary, we choose it such that $\psi(\epsilon)$ is minimum, i.e., $\epsilon =
1/\sqrt{2-\alpha}$; in other words, $\psi(\epsilon) = 3 - \alpha + 2 \sqrt{2 - \alpha} < 6$. This
concludes the proof as the maximum unsuccessful search time is at most twice the maximum cluster
size plus two.
\end{proof}

\begin{Rem}
We have showed that w.h.p.\ the maximum cluster size produced by
\textsc{DecideFirst}$(n,m)$ is in fact not more than
\[ \frac{3 - \alpha + 2 \sqrt{2 - \alpha}}{(1-\alpha)^2} \log_2 \log n + O(1)
    < \frac{6}{(1-\alpha)^2} \log_2 \log n + O(1) \, .
\]
\end{Rem}

\subsection{Two-way Post-linear Probing:
algorithm W{\small{ALK}}F{\small{IRST}}}

We introduce yet another hashing algorithm that achieves $\Theta(\log \log n)$ worst-case search
time, in probability, and shows better performance in experiments than \textsc{DecideFirst}
algorithm as demonstrated in the simulation results presented in Section \ref{SimulationSec}.
Suppose that the load factor $\alpha \in (0, 1/2)$, and that the hash table is divided into blocks
of size
\[\beta_3 (n) := \ceil{\frac{\log_2 \log n + 8}{1-\delta}} \, ,  \]
where $\delta \in (2 \alpha, 1)$ is an arbitrary constant. Define the load of a block to be the
number of keys (or occupied cells) it contains. Suppose that we save with each block its load, and
keep it updated whenever a key is inserted into one of its cells. Recall that each key $x$ has two
initial hashing cells. From these initial cells the algorithm probes linearly and cyclically until
it finds two empty cells $U_x$ and $V_x$, which we call terminal cells. Let $\lambda(U_x)$ and
$\lambda(V_x)$ be the blocks that contain these cells. The algorithm, then, inserts the key $x$
into the terminal cell (among $U_x$ and $V_x$) that belongs to the least loaded block among
$\lambda(U_x)$ and $\lambda(V_x)$, breaking ties randomly. We refer to this algorithm of open
addressing hashing for inserting $m$ keys into $n$ cells as \textsc{WalkFirst}$(n,m)$.

\begin{figure}[htbp]
\begin{center}
  \scalebox{0.88}{\input{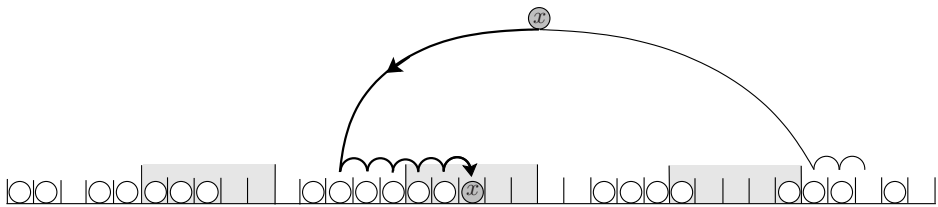}}
  \caption{Algorithm \textsc{WalkFirst}$(n,m)$ inserts each key
into the terminal cell that belongs to the least crowded block, breaking ties
arbitrarily.}\label{figblocking}
\end{center}
\end{figure}

In the remainder of this section, we analyze the worst-case performance of
algorithm \textsc{WalkFirst}$(n,m)$. Recall that the maximum unsuccessful
search time is bounded from above by twice the maximum cluster size plus
two. The following theorem asserts that upon termination of the algorithm,
it is most likely that every block has at least one empty cell. This implies
that the length of the largest cluster is at most $2\beta_3-2$.

\begin{Ther} \label{walkTher}
Let $n \in \mathbb{N}$, and $m = \floor{\alpha n}$, for some constant
$\alpha \in (0,1/2)$. Let $\delta \in (2 \alpha, 1)$ be an arbitrary
constant, and define
\[ \beta_3 (n)
    := \ceil{\frac{\log_2 \log n + 8}{1- \delta}} \, .
\]
Upon termination of algorithm \textsc{WalkFirst}$(n,m)$ with blocks of size
$\beta_3$, the probability that there is a fully loaded block goes to zero
as $n$ tends to infinity. That is, w.h.p., the maximum unsuccessful search
time of \textsc{WalkFirst}$(n,m)$ is at most $4 \beta_3 -2$, and the maximum
insertion time is at most $4 \beta_3 - 4$.
\end{Ther}

For $k \in \intSet{m}$, let us denote by $A_k$ the event that after
the insertion of $k$ keys (i.e., at time $k$), none of the blocks is
fully loaded. To prove Theorem \ref{walkTher}, we shall show that
$\prob{A_m^c}=o(1)$. We do that by using a witness tree argument;
see e.g., \cite{Cole:98a, Cole:98b, Meyer:96, Mitzen:00, Schick:00,
Vocking:99}. We show that if a fully-loaded block exists, then there
is a witness binary tree of height $\beta_3$ that describes the
history of that block. The formal definition of a witness tree is
given below. Let us number the keys $1, \ldots, m$ according to
their insertion time. Recall that each key $t \in \intSet{m}$ has
two initial cells which lead to two terminal empty cells belonging
to two blocks. Let us denote these two blocks available for the
$t$-th key by $X_t$ and $Y_t$. Notice that all the initial cells are
independent and uniformly distributed. However, all terminal
cells---and so their blocks---are not. Nonetheless, for each fixed
$t$, the two random values $X_t$ and $Y_t$ are independent.

\subsubsection*{The History Tree}

We define for each key $t$ a \emph{full history tree} $T_t$ that
describes essentially the history of the block that contains the
$t$-th key up to its insertion time. It is a colored binary tree
that is labelled by key numbers except possibly the leaves, where
each key refers to the block that contains it. Thus, it is indeed a
binary tree that represents all the pairs of blocks available for
all other keys upon which the final position of the key $t$ relies.
Formally, we construct the binary tree node by node in
Breadth-First-Search (\textsc{bfs}) order as follows. First, the
root of $T_t$ is labelled $t$, and is colored white. Any white node
labelled $\tau$ gets two children: a left child corresponding to the
block $X_{\tau}$, and a right child corresponding to the block
$Y_{\tau}$. The left child is labelled and colored according to the
following rules:
\begin{description}
  \item{(a)} If the block $X_{\tau}$ contains some keys at the time of insertion
of key $\tau$, and the last key inserted in that block, say
$\sigma$, has not been encountered thus far in the \textsc{bfs}
order of the binary tree $T_t$, then the node is labelled $\sigma$
and colored white.
  \item{(b)} As in case (a), except that $\sigma$ has already been encountered
in the \textsc{bfs} order. We distinguish such nodes by coloring
them black, but they get the same label $\sigma$.
  \item{(c)} If the block $X_{\tau}$ is empty at the time of insertion of key ${\tau}$,
then it is a ``dead end'' node without any label and it is colored
gray.
\end{description}
Next, the right child of ${\tau}$ is labelled and colored by following the same rules but with the
block $Y_{\tau}$. We continue processing nodes in \textsc{bfs} fashion. A black or gray node in the
tree is a leaf and is not processed any further. A white node with label $\sigma$ is processed in
the same way we processed the key $\tau$, but with its two blocks $X_{\sigma}$ and $Y_{\sigma}$. We
continue recursively constructing the tree until all the leaves are black or gray. See Figure
\ref{fighist1} for an example of a full history tree.

Notice that the full history tree is totally deterministic as it does not contain any random value.
It is also clear that the full history tree contains at least one gray leaf and every internal
(white) node in the tree has two children. Furthermore, since the insertion process is sequential,
node values (key numbers) along any path down from the root must be decreasing (so the binary tree
has the heap property), because any non-gray child of any node represents the last key inserted in
the block containing it at the insertion time of the parent. We will not use the heap property
however.

\begin{figure}[htbp]
\begin{center}
 \scalebox{.8}{\includegraphics{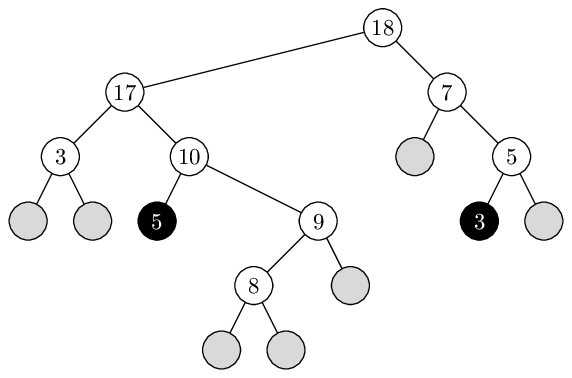}}
 \caption{The full history tree of key 18. White nodes represent type (a)
  nodes. Black nodes are type (b) nodes---they refer to keys already
  encountered in \textsc{bfs} order. Gray nodes are type (c) nodes---they occur
  when a key selects an empty block.}\label{fighist1} 
\end{center}
\end{figure}

Clearly, the full history tree permits one to deduce the load of the block that contains the root
key at the time of its insertion: it is the length of the shortest path from the root to any gray
node. Thus, if the block's load is more than $h$, then all gray nodes must be at distance more than
$h$ from the root. This leads to the notion of a \emph{truncated history tree} of height $h$, that
is, with $h+1$ levels of nodes. The top part of the full history tree that includes all nodes at
the first $h+1$ levels is copied, and the remainder is truncated.

We are in particular interested in truncated history trees without gray nodes. Thus, by the
property mentioned above, the length of the shortest path from the root to any gray node (and as
noted above, there is at least one such node) would have to be at least $h+1$, and therefore, the
load of the block harboring the root's key would have to be at least $h+1$. More generally, if the
load is at least $h + \xi$ for a positive integer $\xi$, then all nodes at the bottom level of the
truncated history tree that are not black nodes (and there is at least one such node) must be white
nodes whose children represent keys that belong to blocks with load of at least $\xi$ at their
insertion time. We redraw these node as boxes to denote the fact that they represent blocks of load
at least $\xi$, and we call them ``block" nodes.

\begin{figure}[htbp]
\begin{center}
 \scalebox{.8}{\includegraphics{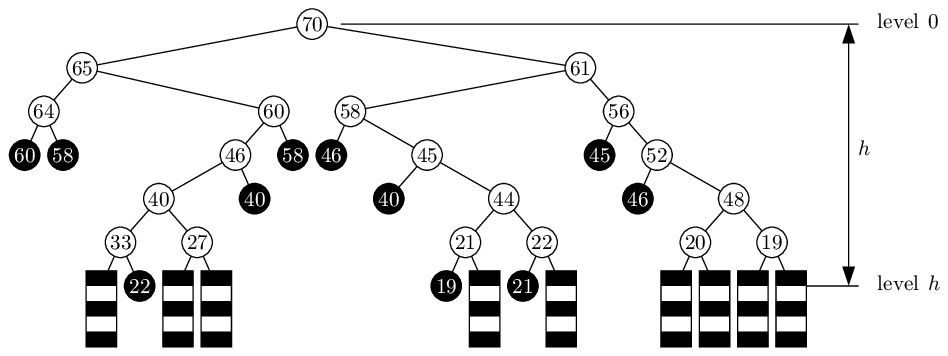}}
 \caption{A witness tree of height $h$ which is a
truncated history tree without gray nodes. The boxes at the lowest
level are block nodes. They represent selected blocks with load of
at least $\xi$. The load of the block that contains key 70 is at
least $h+\xi$.}\label{figwit2}
\end{center}
\end{figure}

\subsubsection*{The Witness Tree}

Let $\xi \in \mathbb{N}$ be a fixed integer to be picked later. For positive integers $h$ and $k$,
where $h + \xi \leq k \leq m$, a \emph{witness tree} $W_k(h)$ is a truncated history tree of a key
in the set $\intSet{k}$, with $h+1$ levels of nodes (thus, of height $h$) and with two types of
leaf nodes, black nodes and ``block" nodes. This means that each internal node has two children,
and the node labels belong to the set $\intSet{k}$. Each black leaf has a label of an internal node
that precedes it in \textsc{bfs} order. Block nodes are unlabelled nodes that represent blocks with
load of at least $\xi$. Block nodes must all be at the furthest level from the root, and there is
at least one such node in a witness tree. Notice that every witness tree is deterministic. An
example of a witness tree is shown in Figure \ref{figwit2}.

Let $\mathcal{W}_k (h,w,b)$ denote the class of all witness trees $W_k(h)$ of height $h$ that have
$w \geq 1$ white (internal) nodes, and $b \leq w$ black nodes (and thus $w-b+1$ block nodes).
Notice that, by definition, the class $\mathcal{W}_k (h,w,b)$ could be empty, e.g., if $w<h$, or $w
\geq 2^h$. However, $\abs{\mathcal{W}_k (h,w,b)} \leq 4^w 2^{w+1} w^b  k^w$, which is due to the
following. Without the labelling, there are at most $4^w$ different shape binary trees, because the
shape is determined by the $w$ internal nodes, and hence, the number of trees is the Catalan number
$\binom{2w}{w}/(w+1) \leq 4^w$. Having fixed the shape, each of the leaves is of one of two types.
Each black leaf can receive one of the $w$ white node labels. Each of the white nodes gets one of
$k$ possible labels.

Note that, unlike the full history tree, not every key has a witness tree $W_k(h)$: the key must be
placed into a block of load of at least $h+\xi-1$ just before the insertion time. We say that a
witness tree $W_k(h)$ \emph{occurs}, if upon execution of algorithm \textsc{WalkFirst}, the random
choices available for the keys represented by the witness tree are actually as indicated in the
witness tree itself. Thus, a witness tree of height $h$ exists if and only if there is a key that
is inserted into a block of load of at least $h+ \xi-1$ before the insertion.

Before we embark on the proof of Theorem \ref{walkTher}, we highlight three important facts whose
proofs are provided in the appendix. First, we bound the probability that a valid witness tree
occurs.

\begin{Lem}\label{WitnessPro}
Let $D$ denote the event that the number of blocks in \textsc{WalkFirst}$(n,m)$ with load of at
least $\xi$, after termination, is at most $n/(a \beta_3 \xi)$, for some constant $a > 0$. For $k
\in \intSet{m}$, let $A_k$ be the event that after the insertion of $k$ keys, none of the blocks is
fully loaded. Then for any positive integers $h, w$ and $k \geq h + \xi$, and a non-negative
integer $b \leq w$, we have
\[ \sup_{W_k(h) \in \mathcal{W}_k(h,w,b)}
        \prob{W_k(h) ~\text{occurs} \; | \; A_{k-1} \cap D }
             \leq \frac{4^w \beta_3^{w+b-1}}{(a \xi)^{w-b+1} n^{w+b-1}} \, .
\]
\end{Lem}

The next lemma asserts that the event $D$ in Lemma \ref{WitnessPro} is most likely to be true, for
sufficiently large $\xi < \beta_3$.

\begin{Lem}\label{highBlocks}
Let $\alpha$, $\delta$, and $\beta_3$ be as defined in Theorem
\ref{walkTher}. Let $N$ be the number of blocks with load of at least $\xi$
upon termination of algorithm \textsc{WalkFirst}$(n,m)$. If $\xi \geq \delta
\beta_3$, then $\prob{N \geq n/(a \beta_3 \xi) } = o(1)$, for any constant
$a>0$.
\end{Lem}

Lemma \ref{ImpLem} addresses a simple but crucial fact. If the height of a witness tree $W_k(h) \in
\mathcal{ W}_k (h,w,b)$ is $h \geq 2$, then the number of white nodes $w$ is at least two, (namely,
the root and its left child); but what can we say about $b$, the number of black nodes?

\begin{Lem} \label{ImpLem}
In any witness tree $W_k(h) \in \mathcal{ W}_k (h,w,b)$, if $h \geq 2$ and $ w \leq 2^{h- \eta}$,
where $\eta \geq 1$, then the number $b$ of black nodes  is $ \geq \eta $, i.e., $\ind{[b \geq
\eta] \cup [w>2^{h- \eta}]} = 1$.
\end{Lem}

\paragraph{Proof of Theorem \ref{walkTher}.}
Recall that $A_k$, for $k \in \intSet{m}$, is the event that after
the insertion of $k$ keys (i.e., at time $k$), none of the blocks is
fully loaded. Notice that $A_m \subseteq A_{m-1} \subseteq \cdots
\subseteq A_1$, and the event $A_{\beta_3-1}$ is deterministically
true. We shall show that $\prob{A_m^c} = o(1)$. Let $D$ denote the
event that the number of blocks with load of at least $\xi$, after
termination, is at most $n/(a \beta_3 \xi)$, for some constant $a
>1$ to be decided later. Observe that
\begin{eqnarray*}
\prob{A_m^c}
    &\leq& \prob{D^c} +  \prob{A_m^c \, | \, D } \\
    &\leq& \prob{D^c} + \prob{A_m^c \, | \, A_{m-1} \cap D}
                + \prob{A_{m-1}^c \, | \, D}  \\
    &\vdots&   \\
    &\leq& \prob{D^c} +
               \sum_{k=\beta_3}^m \prob{A_k^c \, | \, A_{k-1} \cap D} \, .
\end{eqnarray*}
Lemma \ref{highBlocks} reveals that $\prob{D^c}=o(1)$, and hence, we only need to demonstrate that
$p_k := \prob{A_k^c \, | \, A_{k-1} \cap D} = o(1/n)$, for $k = \beta_3, \ldots, m$. We do that by
using the witness tree argument. Let $h, \xi, \eta \in [2, \infty)$ be some integers to be picked
later such that $h + \xi \leq \beta_3$. If after the insertion of $k$ keys, there is a block with
load of at least $h+ \xi$, then a witness tree $W_k(h)$ (with block nodes representing blocks with
load of at least $\xi$) must have occurred. Recall that the number of white nodes $w$ in any
witness tree $W_k(h)$ is at least two. Using Lemmas \ref{WitnessPro} and \ref{ImpLem}, we see that
\begin{eqnarray*}
p_k
  &\leq& \sum_{W_k(h)} \prob{ W_k(h) ~\text{occurs}
                             \, | \, A_{k-1} \cap D } \\
  &\leq& \sum_{w=2}^{ 2^h-1} \sum_{b=0}^w ~
        \sum_{W_k(h) \in \mathcal{ W}_k (h,w,b)}
               \prob{ W_k(h) ~\text{occurs} \, | \, A_{k-1}  \cap D} \\
  &\leq& \sum_{w=2}^{ 2^h-1}  \sum_{b=0}^w \, \abs{\mathcal{W}_k (h,w,b)}
          \, \sup_{W_k(h) \in \mathcal{ W}_k (h,w,b)}
              \prob{ W_k(h) ~\text{occurs} \, | \, A_{k-1}  \cap D} \\
 &\leq& \sum_{w=2}^{2^h} \sum_{b=0}^w
          \frac{2^{w+1} 4^{2w} w^b  k^w \beta_3^{w+b-1} }{ (a \xi)^{w-b+1}
          \, n^{w+b-1} } \ind{[b \geq \eta] \cup [w>2^{h- \eta}]} \\
  &\leq& \frac{2n}{a \xi \beta_3} \sum_{w=2}^{2^h}
        \pran{ \frac{32 \alpha  \beta_3}{a \xi} }^w
          \sum_{b=0}^w \pran{\frac{a w \xi \beta_3}{n}}^b
              \ind{[b \geq \eta] \cup [w>2^{h- \eta}]} \, .
\end{eqnarray*}
Note that we disallow $b=w+1$, because any witness tree has at least one block node. We split the
sum over $w \leq 2^{h- \eta}$, and $w > 2^{h- \eta}$. For $w \leq 2^{h- \eta}$, we have $b \geq
\eta$, and thus
\begin{eqnarray*}
\sum_{b =0}^w \pran{\frac{a w \xi \beta_3}{n}}^b
              \ind{[b \geq \eta] \cup  [w>2^{h- \eta}]}
  &=&   \sum_{b= \eta}^w \pran{\frac{a w \xi \beta_3}{n}}^b \\
 &\leq&  \pran{\frac{a w \xi \beta_3}{n}}^{\eta} \sum_{b=0}^{\infty}
                    \pran{\frac{a w \xi \beta_3}{n}}^b \\
 &<&  2 \pran{\frac{a w \xi \beta_3}{n}}^{\eta} \, ,
\end{eqnarray*}
provided that $n$ is so large that $a 2^{h+1} \xi \beta_3 \leq n$, (this insures that $a w \xi
\beta_3/n < 1/2$). For $w \in (2^{h- \eta}, \, 2^h]$, we bound trivially, assuming the same large
$n$ condition:
\[  \sum_{b=0}^w \pran{ \frac{a w \xi \beta_3}{n}}^b \leq 2 \, .  \]
In summary, we see that
\[  p_k \leq 4n  \sum_{ w > 2^{h- \eta}}
           \pran{\frac{32 \alpha  \beta_3}{a \xi} }^w +
    4 \pran{\frac{a \xi \beta_3}{ n}}^{\eta -1} \, \sum_{w=2}^{2^{h- \eta} }
               \pran{\frac{32 \alpha  \beta_3 }{a \xi} }^w  w^{\eta} \, .
\]
We set $a=32$, and $\xi = \ceil{ \delta \beta_3}$, so that $32 \alpha
\beta_3 /(a \xi) \leq 1/2$, because $\delta \in (2 \alpha, 1)$. With this
choice, we have
\[  p_k \leq \frac{ 4n }{2^{2^{h- \eta}} }
        + 4c \pran{\frac{32 \beta_3^2}{n}}^{\eta -1}       \, ,  \]
where $c = \sum_{w \geq 2} w^{\eta} /2^w$. Clearly, if we put $h = \eta + \ceil{ \log_2 \log_2
n^{\eta} }$, and $\eta=3$, then we see that $h+ \xi \leq \beta_3$, and $p_k = o(1/n)$. Notice that
$h$ and $\xi$ satisfy the technical condition $a 2^{h+1} \xi \beta_3 \leq n$, asymptotically.
\hfill$\square$

\begin{Rem}
The restriction on $\alpha$ is needed only to prove Lemma \ref{highBlocks} where the binomial tail
inequality is valid only if $\alpha < 1/2$. Simulation results, as we show next, suggest that a
variant of Theorem \ref{walkTher} might hold for any $\alpha \in (0,1)$ with block  size
$\floor{(1-\alpha )^{-1} \log_2 \log n}$.
\end{Rem}

\subsection*{Tradeoffs}

We have seen that by using two linear probe sequences instead of just one, the maximum unsuccessful
search time decreases exponentially from $O(\log n)$ to $O(\log \log n)$. The average search time,
however, could at worst double as shown in the simulation results. Most of the results presented in
this article can be improved, by a constant factor though, by increasing the number of hashing
choices per key. For example, Theorems \ref{universalLower1} and \ref{universalLower2} can be
easily generalized for open addressing hashing schemes that use $d \geq 2$ linear probe sequences.
Similarly, all the two-way linear probing algorithms we design here can be generalized to $d$-way
linear probing schemes. The maximum unsuccessful search time will, then, be at most $d \, C \log_d
\log n + O(d)$, where $C$ is a constant depending on $\alpha$. This means that the best worst-case
performance is when $d=3$ where the minimum of $d / \log d$ is attained. The average search time,
on the other hand, could triple.

The performance of these algorithms can be further improved by using V\"ocking's scheme
\textsc{LeftMC}$(n, m, d)$, explained in Section \ref{MultipleSec}, with $d \geq 2$ hashing
choices. The maximum unsuccessful search time, in this case, is at most $C \log \log n / \log
\phi_d + O(d)$, for some constant $C$ depending on $\alpha$. This is minimized when $d=o(\log \log
n)$, but we know that it can not get better than $C\log_2 \log n + O(d)$, because $\lim_{d \to
\infty} \phi_d = 2$.

\section{Simulation Results} \label{SimulationSec}

We simulate all linear probing algorithms we discussed in this article with the \textsc{fcfs}
replacement strategy: the classical linear probing algorithm \textsc{ClassicLinear}, the locally
linear algorithm \textsc{LocallyLinear}, and the two-way linear probing algorithms
\textsc{ShortSeq}, \textsc{SmallCluster}, \textsc{WalkFirst}, and \textsc{DecideFirst}. For each
value of $n \in \set{2^8, 2^{12}, 2^{16}, 2^{20}, 2^{22}}$, and constant $\alpha \in \set{0.4,
0.9}$, we simulate each algorithm 1000 times divided into 10 iterations (experiments). Each
iteration consists of 100 simulations of the same algorithm where we insert $\floor{\alpha n}$ keys
into a hash table with $n$ cells. In each simulation we compute the average and the maximum
successful search and insert times. For each iteration (100 simulations), we compute the average of
the average values and and the average of the maximum values computed during the 100 simulations
for the successful search and insert times. The overall results are finally averaged over the 10
iterations and recorded in the next Tables. Similarly, the average maximum cluster size is computed
for each algorithm as it can be used to bound the maximum unsuccessful search time, as mentioned
earlier. Notice that in the case of the algorithms \textsc{ClassicLinear} and \textsc{ShortSeq},
the successful search time is the same as the insertion time.

Tables \ref{SimData1} and \ref{SimData2} contain the simulation results of the algorithms
\textsc{ClassicLinear}, \textsc{ShortSeq}, and \textsc{SmallCluster}. With the exception of the
average insertion time of \textsc{SmallCluster} Algorithm, which is slightly bigger than
\textsc{ClassicLinear} Algorithm, it is evident that the average and the worst-case performances of
\textsc{SmallCluster} and \textsc{ShortSeq} are better than \textsc{ClassicLinear}. Algorithm
\textsc{SmallCluster} seems to have the best worst-case performance among the three algorithms.
This is not a total surprise to us, because the algorithm considers more information (relative to
the other two) before it makes its decision of where to insert the keys. It is also clear that
there is a nonlinear increase, as a function of $n$, in the difference between the performances of
these algorithms. This may suggest that the worst-case performances of algorithms \textsc{ShortSeq}
and \textsc{SmallCluster} are roughly of the order of $\log n$.

\begin{table}[htbp]
 \centering
\small{ \begin{tabular}{|c|c||c|r||c|r||c|r||c|r|}  \hline
\multirow{3}{*}{$n$} &  \multirow{3}{*}{$\alpha$} & \multicolumn{2}{|c||}{\textsc{ClassicLinear}}& \multicolumn{2}{c||}{\textsc{ShortSeq}} & \multicolumn{2}{|c||}{\textsc{SmallCluster}} & \multicolumn{2}{|c|}{\textsc{SmallCluster} }\\[-5pt]
& & \multicolumn{2}{|c||}{\scriptsize{Insert/Search Time}}& \multicolumn{2}{c||}{\scriptsize{Insert/Search Time}} & \multicolumn{2}{|c||}{\scriptsize{Search Time}} & \multicolumn{2}{|c|}{\scriptsize{Insert Time} }\\
\cline{3-10}
   &  &  \makebox[9mm]{Avg} & Max  &  \makebox[9mm]{Avg} & Max & \makebox[8mm]{Avg}  & Max& \makebox[9mm]{Avg}  & Max\\    \hline \hline
 \multirow{2}{*}{$2^8$}
    & 0.4  & 1.33  & 5.75 & \textbf{\textit{1.28}} & \textit{4.57} & 1.28 &  \textbf{4.69} & 1.50 & 9.96\\[-3pt] \cline{2-2}
    & 0.9  & 4.38 & 68.15 & \textbf{\textit{2.86}} & \textit{39.72} & 3.05 & \textbf{35.69} & 6.63 & 71.84\\[-3pt] \cline{1-2}
 \multirow{2}{*}{$2^{12}$}
    & 0.4  & 1.33 & 10.66 & \textbf{\textit{1.28}} & \textit{7.35} & 1.29 & \textbf{7.49} & 1.52 & 14.29\\[-3pt] \cline{2-2}
    & 0.9  & 5.39 & 275.91 & \textbf{\textit{2.90}} & \textit{78.21} & 3.07 & \textbf{66.03} & 6.91 & 118.34 \\[-3pt]  \cline{1-2}
 \multirow{2}{*}{$2^{16}$}
    & 0.4  & 1.33 & 16.90 & \textbf{\textit{1.28}} & \textit{10.30} & 1.29 &  \textbf{10.14} & 1.52 & 18.05\\[-3pt] \cline{2-2}
    & 0.9  & 5.49 & 581.70 & \textbf{\textit{2.89}} & \textit{120.32} & 3.07 & \textbf{94.58}  & 6.92 & 155.36\\[-3pt]  \cline{1-2}
 \multirow{2}{*}{$2^{20}$}
    & 0.4  & 1.33 & 23.64 & \textbf{\textit{1.28}} & \textit{13.24} & 1.29 &  \textbf{13.03} & 1.52 & 21.41\\[-3pt]  \cline{2-2}
    & 0.9  & 5.50 & 956.02 & \textbf{\textit{2.89}} & \textit{164.54} & 3.07 & \textbf{122.65}  & 6.92 & 189.22 \\[-3pt]\cline{1-2}
 \multirow{2}{*}{$2^{22}$}
    & 0.4  & 1.33 & 26.94 & \textbf{\textit{1.28} }&\textit{14.94 } & 1.29 &  \textbf{14.44} & 1.52 & 23.33 \\[-3pt]  \cline{2-2}
    & 0.9  & 5.50 & 1157.34 & \textbf{\textit{2.89}} & \textit{188.02} & 3.07 & \textbf{136.62} & 6.93 &  205.91  \\ \hline
\end{tabular}}
  \caption{The average and the maximum successful search and insert times averaged over 10
iterations each consisting of 100 simulations of the algorithms. The best successful search time is
shown in boldface and the best insert time is shown in italic.}\label{SimData1}
\end{table}

\begin{table}[htbp]
 \centering
\small{  \begin{tabular}{|c|c||r|r||r|r||r|r|}  \hline
\multirow{2}{*}{$n$} &  \multirow{2}{*}{$\alpha$} & \multicolumn{2}{|c||}{\textsc{ClassicLinear}}& \multicolumn{2}{c||}{\textsc{ShortSeq}} & \multicolumn{2}{|c|}{\textsc{SmallCluster}} \\ \cline{3-8}
   &  &  \makebox[9mm]{Avg} & Max  &  \makebox[9mm]{Avg} & Max & \makebox[9mm]{Avg}  & Max\\    \hline \hline
 \multirow{2}{*}{$2^8$}
    & 0.4  &2.02 & 8.32   & 1.76   & 6.05  & \textbf{1.76} & \textbf{5.90}   \\[-3pt] \cline{2-2}
    & 0.9  & 15.10 & 87.63 & 12.27  & 50.19 & \textbf{12.26}& \textbf{43.84} \\[-3pt]  \cline{1-2}
 \multirow{2}{*}{$2^{12}$}
    & 0.4  & 2.03 & 14.95 & 1.75    & 9.48  &\textbf{1.75} & \textbf{9.05}  \\[-3pt] \cline{2-2}
    & 0.9  &15.17 & 337.22 & 12.35 & 106.24                   & \textbf{12.34}& \textbf{78.75} \\[-3pt]  \cline{1-2}
 \multirow{2}{*}{$2^{16}$}
    & 0.4  & 2.02 & 22.54 & 1.75    & 12.76 & \textbf{1.75} & \textbf{12.08} \\[-3pt] \cline{2-2}
    & 0.9  & 15.16 & 678.12 & 12.36 & 155.26   &\textbf{12.36}& \textbf{107.18}                   \\[-3pt]  \cline{1-2}
 \multirow{2}{*}{$2^{20}$}
    & 0.4  & 2.02 & 29.92 & 1.75    & 16.05 &\textbf{1.75} & \textbf{15.22} \\[-3pt]  \cline{2-2}
    & 0.9  & 15.17 & 1091.03    & 12.35   & 203.16 & \textbf{12.35}& \textbf{136.19}                   \\[-3pt] \cline{1-2}
 \multirow{2}{*}{$2^{22}$}
    & 0.4  & 2.02 & 33.81 & 1.75    & 17.74 &\textbf{1.75} & \textbf{16.65} \\[-3pt]  \cline{2-2}
    & 0.9  & 15.17 & 1309.04  & 12.35   & 226.44   & \textbf{12.35}& \textbf{150.23} \\ \hline
\end{tabular}}
  \caption{The average maximum cluster size and the average cluster size
over 100 simulations of the algorithms. The best performances are drawn in
boldface.}\label{SimData2}
\end{table}

The simulation data of algorithms \textsc{LocallyLinear}, \textsc{WalkFirst}, and
\textsc{DecideFirst} are presented in Tables \ref{SimData3}, \ref{SimData4} and \ref{SimData5}.
These algorithms are simulated with blocks of size $\floor{(1-\alpha )^{-1} \log_2 \log n}$. The
purpose of this is to show that, practically, the additive and the multiplicative constants
appearing in the definitions of the block sizes stated in Theorems \ref{localLinearTher},
\ref{decideTher} and \ref{walkTher} can be chosen to be small. The hash table is partitioned into
equal-sized blocks, except possibly the last one. The average and the maximum values of the
successful search time, inset time, and cluster size (averaged over 10 iterations each consisting
of 100 simulations of the algorithms) are recorded in the tables below where the best performances
are drawn in boldface.

Results show that  \textsc{LocallyLinear} Algorithm has the best performance; whereas algorithm
\textsc{WalkFirst} appears to perform better than \textsc{DecideFirst}. Indeed, the sizes of the
cluster produced by algorithm \textsc{WalkFirst} appears to be very close to that of
\textsc{LocallyLinear} Algorithm.  This supports the conjecture that Theorem \ref{walkTher} is, in
fact, true for any constant load factor $\alpha \in (0,1)$, and the maximum unsuccessful search
time of \textsc{WalkFirst} is at most $4 (1-\alpha)^{-1} \log_2 \log n + O(1)$, w.h.p. The average
maximum cluster size of algorithm \textsc{DecideFirst} seems to be close to the other ones when
$\alpha$ is small; but it almost doubles when $\alpha$ is large. This may suggest that the
multiplicative constant in the maximum unsuccessful search time established in Theorem
\ref{decideTher} could be improved.

\begin{table}[htbp]
 \centering  \addtolength{\tabcolsep}{-2pt}
\small{  \begin{tabular}{|c|c||r|r||r|r||r|r|}  \hline
\multirow{2}{*}{$n$} &  \multirow{2}{*}{$\alpha$} &   \multicolumn{2}{c||}{\textsc{LocallyLinear}} &  \multicolumn{2}{c||}{\textsc{WalkFirst}} & \multicolumn{2}{c|}{\textsc{DecideFirst}}\\ \cline{3-8}
   &  &  \makebox[9mm]{Avg} & \makebox[9mm]{Max}  &  \makebox[9mm]{Avg} & \makebox[9mm]{Max} & \makebox[9mm]{Avg}  & \makebox[9mm]{Max}\\   \hline  \hline
 \multirow{2}{*}{$2^8$}
    & 0.4   & \textbf{1.73} & \textbf{4.73}& 1.78& 5.32& 1.75& 5.26  \\[-3pt]
  \cline{2-2}
    & 0.9   & \textbf{4.76} &\textbf{36.23} & 4.76& 43.98& 5.06 & 59.69   \\[-3pt]\cline{1-2}
  \multirow{2}{*}{$2^{12}$}
    & 0.4   & \textbf{1.74} & \textbf{6.25}&1.80 &7.86 & 1.78& 7.88  \\[-3pt]\cline{2-2}
    & 0.9   & \textbf{4.76} &\textbf{ 47.66}& 4.80& 67.04& 4.94& 108.97  \\[-3pt]\cline{1-2}
 \multirow{2}{*}{$2^{16}$}
    & 0.4   &\textbf{1.76} & \textbf{7.93}& 1.80& 9.84& 1.78&  10.08 \\[-3pt] \cline{2-2}
    & 0.9   & \textbf{4.78} & \textbf{56.40}& 4.89& 89.77& 5.18&137.51   \\[-3pt] \cline{1-2}
 \multirow{2}{*}{$2^{20}$}
    & 0.4   & \textbf{1.76}& \textbf{8.42}& 1.81& 12.08& 1.79& 12.39  \\[-3pt]  \cline{2-2}
    & 0.9   & \textbf{4.77} & \textbf{65.07}& 4.98& 108.24& 5.26& 162.04  \\[-3pt]\cline{1-2}
 \multirow{2}{*}{$2^{22}$}
    & 0.4   & \textbf{1.76}&  \textbf{9.18}& 1.81& 12.88&  1.79&13.37 \\[-3pt] \cline{2-2}
    & 0.9   & \textbf{4.80}& \textbf{71.69}& 5.04&118.06 &5.32 & 181.46  \\  \hline
\end{tabular}}
\addtolength{\tabcolsep}{2pt}
  \caption{The average and the maximum successful search time averaged over 10 iterations
each consisting of 100 simulations of the algorithms. The best performances are drawn in
boldface.}\label{SimData3}
\end{table}

\begin{table}[htbp]
 \centering  \addtolength{\tabcolsep}{-2pt}
\small{  \begin{tabular}{|c|c||r|r||r|r||r|r|}  \hline
\multirow{2}{*}{$n$} &  \multirow{2}{*}{$\alpha$} &   \multicolumn{2}{c||}{\textsc{LocallyLinear}} &  \multicolumn{2}{c||}{\textsc{WalkFirst}} & \multicolumn{2}{c|}{\textsc{DecideFirst}}\\ \cline{3-8}
   &  &  \makebox[9mm]{Avg} & \makebox[9mm]{Max}  &  \makebox[9mm]{Avg} & \makebox[9mm]{Max} & \makebox[9mm]{Avg}  & \makebox[9mm]{Max}\\   \hline  \hline
 \multirow{2}{*}{$2^8$}
    & 0.4   & \textbf{1.14 }& \textbf{2.78 }&  2.52 &  6.05 & 1.15  & 3.30\\[-3pt]
  \cline{2-2}
    & 0.9   & \textbf{2.89}& \textbf{22.60} & 6.19 &48.00 &3.19 &42.64   \\[-3pt]\cline{1-2}
  \multirow{2}{*}{$2^{12}$}
    & 0.4   & \textbf{1.14}&\textbf{3.38} &2.53 & 8.48&1.17 &5.19   \\[-3pt]\cline{2-2}
    & 0.9   & \textbf{2.91}& \textbf{27.22}& 6.28&69.30 &3.16 &84.52   \\[-3pt]\cline{1-2}
 \multirow{2}{*}{$2^{16}$}
    & 0.4   & \textbf{1.15}& \textbf{4.08}& 2.53& 10.40&1.17 &6.56   \\[-3pt] \cline{2-2}
    & 0.9   & \textbf{2.84}& \textbf{31.21} & 6.43 &91.21 &3.17 &106.09   \\[-3pt] \cline{1-2}
 \multirow{2}{*}{$2^{20}$}
    & 0.4   &\textbf{ 1.15}& \textbf{4.64}& 2.54& 12.58& 1.18&  8.16 \\[-3pt]  \cline{2-2}
    & 0.9   &\textbf{2.89} &\textbf{35.21} & 6.54& 109.71& 3.22&  117.42 \\[-3pt]\cline{1-2}
 \multirow{2}{*}{$2^{22}$}
    & 0.4   & \textbf{1.15}& \textbf{4.99}& 2.54& 13.41& 1.18& 8.83  \\[-3pt] \cline{2-2}
    & 0.9   & \textbf{2.91}&\textbf{38.75} & 6.61& 119.07& 3.26&  132.83 \\  \hline
\end{tabular}}
\addtolength{\tabcolsep}{2pt}
  \caption{The average and the maximum insert time averaged over 10 iterations each consisting
of 100 simulations of the algorithms. The best performances are drawn in boldface.}\label{SimData4}
\end{table}

\begin{table}[htbp]
 \centering  \addtolength{\tabcolsep}{-2pt}
\small{  \begin{tabular}{|c|c||r|r||r|r||r|r|}  \hline
\multirow{2}{*}{$n$} &  \multirow{2}{*}{$\alpha$} &   \multicolumn{2}{c||}{\textsc{LocallyLinear}} &  \multicolumn{2}{c||}{\textsc{WalkFirst}} & \multicolumn{2}{c|}{\textsc{DecideFirst}}\\ \cline{3-8}
   &  &  \makebox[9mm]{Avg} & \makebox[9mm]{Max}  &  \makebox[9mm]{Avg} & \makebox[9mm]{Max} & \makebox[9mm]{Avg}  & \makebox[9mm]{Max}\\   \hline  \hline
 \multirow{2}{*}{$2^8$}
    & 0.4   & \textbf{1.57} & \textbf{4.34}   &  1.65   & 4.70    & 1.63    &  4.81     \\[-3pt]
  \cline{2-2}
    & 0.9   & \textbf{12.18}   & \textbf{33.35}     & 12.54   &  34.40      & 13.48  &  47.76    \\[-3pt] \cline{1-2}
  \multirow{2}{*}{$2^{12}$}
    & 0.4   &\textbf{1.62} & \textbf{6.06}         & 1.68   &  6.32  &  1.68        & 6.82  \\[-3pt] \cline{2-2}
    & 0.9   & \textbf{12.42} & \textbf{48.76}     & 12.78                       &  51.80 &  13.45     & 94.98                 \\[-3pt] \cline{1-2}
 \multirow{2}{*}{$2^{16}$}
    & 0.4   & \textbf{1.62}    & \textbf{7.14}    & 1.68 &  7.31 &  1.68      &  8.92        \\ [-3pt] \cline{2-2}
    & 0.9   & \textbf{12.66}   & \textbf{59.61}   & 12.98    &  62.24  & 13.53  &  125.40     \\[-3pt] \cline{1-2}
 \multirow{2}{*}{$2^{20}$}
    & 0.4   & \textbf{1.65}    & \textbf{8.25}     & 1.71   &  8.50  & 1.71      &  10.76      \\[-3pt]  \cline{2-2}
    & 0.9   & \textbf{12.83}  & \textbf{67.23}     & 13.11  &  69.45  & 13.62    & 145.30       \\[-3pt]\cline{1-2}
 \multirow{2}{*}{$2^{22}$}
    & 0.4   & \textbf{1.62}      & \textbf{8.90}   & 1.71  &  8.95 & 1.71   & 11.46     \\[-3pt]  \cline{2-2}
    & 0.9   & \textbf{12.72}  & \textbf{65.58}    & 13.19&  73.22  &  13.66   &  164.45      \\ \hline
\end{tabular}}
\addtolength{\tabcolsep}{2pt}
  \caption{The average and the maximum cluster sizes averaged over 10 iterations each consisting
of 100 simulations of the algorithms. The best performances are drawn in boldface.}\label{SimData5}
\end{table}

Comparing the simulation data from all tables, one can see that the best average performance is
achieved by the algorithms \textsc{LocallyLinear} and \textsc{ShortSeq}. Notice that
\textsc{ShortSeq} Algorithm achieves the best average successful search time when $\alpha=0.9$. The
best (average and maximum) insertion time is achieved by the algorithm \textsc{LocallyLinear}. On
the other hand, algorithms \textsc{WalkFirst} and \textsc{LocallyLinear} are superior to the others
in worst-case performance. It is worth noting that surprisingly, the worst-case successful search
time of algorithm \textsc{SmallCluster} is very close to the one achieved by \textsc{WalkFirst} and
better than that of \textsc{DecideFirst}, although, it appears that the difference becomes larger,
as $n$ increases.

\section*{Appendix}

For completeness, we prove the lemmas used in the proof of Theorem \ref{walkTher}.

\begin{Lem4}
Let $D$ denote the event that the number of blocks in \textsc{WalkFirst}$(n,m)$ with load of at
least $\xi$, after termination, is at most $n/(a \beta_3 \xi)$, for some constant $a > 0$. For $k
\in \intSet{m}$, let $A_k$ be the event that after the insertion of $k$ keys, none of the blocks is
fully loaded. Then for any positive integers $h, w$ and $k \geq h + \xi$, and a non-negative
integer $b \leq w$, we have
\[ \sup_{W_k(h) \in \mathcal{W}_k(h,w,b)}
        \prob{W_k(h) ~\text{occurs} \; | \; A_{k-1} \cap D }
             \leq \frac{4^w \beta_3^{w+b-1}}{(a \xi)^{w-b+1} n^{w+b-1}} \, .
\]
\end{Lem4}
\begin{proof}
Notice first that given $A_{k-1}$, the probability that any fixed key in the set $\intSet{k}$
chooses a certain block is at most $2\beta_3/n$.  Let $W_k(h) \in \mathcal{ W}_k (h,w,b)$ be a
fixed witness tree. We compute the probability that $W_k(h)$ occurs given that $A_{k-1}$ is true,
by looking at each node in \textsc{bfs} order. Suppose that we are at an internal node, say $u$, in
$W_k(h)$. We would like to find the conditional probability that a certain child of node $u$ is
exactly as indicated in the witness tree, given that $A_{k-1}$ is true, and everything is revealed
except those nodes that precede $u$ in the \textsc{bfs} order. This depends on the type of the
child. If the child is white or black, the conditional probability is not more than $2 \beta_3/n$.
This is because each key refers to the unique block that contains it, and moreover, the initial
hashing cells of all keys are independent. Multiplying just these conditional probabilities yields
$(2\beta_3/n)^{w+b-1}$, as there are $w+b-1$ edges in the witness tree that have a white or black
nodes as their lower endpoint. On the other hand, if the child is a block node, the conditional
probability is at most $2/(a \xi)$. This is because a block node corresponds to a block with load
of at least $\xi$, and there are at most $n/(a \beta_3 \xi)$ such blocks each of which is chosen
with probability of at most $2 \beta_3/n$. Since there are $w-b+1$ block nodes, the result follows
plainly by multiplying all the conditional probabilities.
\end{proof}

To prove Lemma \ref{highBlocks}, we need to recall the following binomial tail inequality
\cite{Okamoto:58}: for $p \in (0,1)$, and any positive integers $r$, and $t \geq \eta r p$, for
some $\eta>1$, we have
\[ \prob{\binomial(r,p) \geq t}  \leq \pran{\varphi \pran{\frac{t}{rp}} }^t
               \leq (\varphi(\eta))^t  \, ,  \]
where $\varphi(x) = x^{-1} e^{1-1/x}$, which is decreasing on $(1, \infty)$. Notice that
$\varphi(x) < 1$, for any $x > 1$, because $1/x=(1-z) < e^{-z} = e^{1/x-1}$, for some $z \in
(0,1)$.

\begin{Lem5}
Let $\alpha$, $\delta$, and $\beta_3$ be as defined in Theorem \ref{walkTher}. Let $N$ be the
number of blocks with load of at least $\xi$ upon termination of algorithm
\textsc{WalkFirst}$(n,m)$. If $\xi \geq \delta \beta_3$, then $\prob{N \geq n/(a \beta_3 \xi) } =
o(1)$, for any constant $a>0$.
\end{Lem5}
\begin{proof}
Fix $\xi \geq \delta \beta_3$. Let $B$ denote the last block in the hash table, i.e., $B$ consists
of the cells $n-\beta_3, \ldots, n-1$. Let $L$ be the load of $B$ after termination. Since the
loads of the blocks are identically distributed, we have
\[ \Exp{N} = \frac{n}{\beta_3} \prob{L \geq \xi} \, .  \]
Let $S$ be the set of the consecutively occupied cells, after termination, that occur between the
first empty cell to the left of the block $B$ and the cell $n- \beta_3$; see Figure \ref{defBS}.

\begin{figure}[htbp]
\begin{center}
  \scalebox{1}{\input{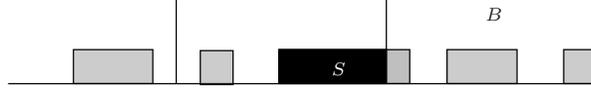}}
  \caption{The last part of the hash table showing clusters, the last
block $B$, and the set $S$.}\label{defBS}
\end{center}
\end{figure}

We say that a key is born in a set of cells $A$ if at least one of its two initial hashing cells
belong to $A$. For convenience, we write $\nu(A)$ to denote the number of keys that are born in
$A$. Obviously, $\nu(A)$ is $\binomial(m, 2\abs{A}/n)$. Since the cell adjacent to the left
boundary of $S$ is empty, all the keys that are inserted in $S$ are actually born in $S$. That is,
if $\abs{S} = j$, then $\nu(S) \geq j$. So, by the binomial tail inequality given earlier, we see
that
\[ \prob{\abs{S} =j }
     =   \prob{\event{\nu(S) \geq j} \cap \event{\abs{S} =j }}
   \leq  \prob{\binomial(m, 2j/n) \geq j} \leq c^j \, ,
\]
where the constant $c := \varphi(1/(2 \alpha)) = 2 \alpha e^{1-2 \alpha} < 1$, because $\alpha <
1/2$. Let
\[ \ell :=  \log_c \frac{1-c}{\xi^2} = O(\log \beta_3) \, . \]
and notice that for $n$ large enough,
\[ \xi \geq \delta \beta_3
       \geq \frac{\delta 2m (\ell+ \beta_3)}{(1+\ell/\beta_3)2 \alpha n}
       \geq  y \frac{ 2m(\ell+\beta_3)}{n}   \, ,
\]
where $y = 1/2 + \delta /(4 \alpha) > 1$, because $\delta \in (2 \alpha, 1)$. Clearly, by the same
property of $S$ stated above, $L \leq \nu(S \cup B)$; and hence, by the binomial tail inequality
again, we conclude that for $n$ sufficiently large,
\begin{eqnarray*}
 \prob{L \geq \xi }
   &\leq& \prob{\event{\nu(S \cup B) \geq \xi}
                          \cap \event{\abs{S} \leq \ell}}
              + \sum_{j=\ell}^m \prob{\abs{S} = j } \\
   &\leq& \prob{ \binomial(m, 2(\ell+\beta_3)/n) \geq \xi } + \frac{c^\ell}{1-c} \\
   &\leq& \pran{\varphi(y)}^\xi   + \frac{c^\ell}{1-c}
   ~ \leq ~ \frac{1}{\xi^2} + \frac{1}{\xi^2} ~ = ~ \frac{2}{\xi^2}\, .
\end{eqnarray*}
Thence, $\Exp{N} \leq 2 n/(\beta_3 \xi^2)$ which implies by Markov's inequality that
\[ \prob{N \geq \frac{n}{a \beta_3 \xi}} \leq \frac{2 a}{\xi} = o(1) \, . \]
\end{proof}

\begin{Lem6}
In any witness tree $W_k(h) \in \mathcal{ W}_k (h,w,b)$, if $h \geq 2$ and $ w \leq 2^{h- \eta}$,
where $\eta \geq 1$, then the number $b$ of black nodes  is $ \geq \eta $, i.e., $\ind{[b \geq
\eta] \cup [w>2^{h- \eta}]} = 1$.
\end{Lem6}
\begin{proof}
Note that any witness tree has at least one block node at distance $h$ from the root. If we have
$b$ black nodes, the number of block nodes is at least $2^{h-b}$. Since $w \leq 2^{h- \eta}$, then
$2^{h- \eta} -b +1 \geq w-b +1 \geq 2^{h-b}$. If $b=0$, then we have a contradiction. So, assume $b
\geq 1$. But then $2^{h- \eta} \geq 2^{h-b}$; that is, $b \geq \eta$.
\end{proof}

\end{document}